\documentclass{article}

 \usepackage[preprint]{neurips_2026}

\usepackage[utf8]{inputenc} 
\usepackage[T1]{fontenc}    
\usepackage[table]{xcolor}
\definecolor{dark-blue}{RGB}{0,0,191}
\usepackage[hidelinks,colorlinks=true,citecolor=dark-blue,linkcolor=dark-blue,urlcolor=dark-blue, backref=page]{hyperref}
\usepackage{url}            
\usepackage{booktabs}       
\usepackage{amsfonts}       
\usepackage{nicefrac}       
\usepackage{microtype}      
\usepackage{xcolor}         

\title{Fast Rates in $\alpha$-Potential Games via Regularized Mirror Descent}

\usepackage{pifont}

\usepackage{amsmath}
\usepackage{amssymb}
\usepackage{mathtools}
\usepackage{amsthm}
\usepackage{color}
\mathtoolsset{showonlyrefs}

\usepackage{algorithm}
\usepackage{algpseudocode}

\usepackage{enumitem}

\usepackage[capitalize,noabbrev]{cleveref}

\usepackage{xurl} 
\usepackage{listings}
\usepackage{xcolor}
\usepackage{thmtools}

\usepackage{pifont}
\theoremstyle{plain}
\newtheorem{theorem}{Theorem}[section]

\newtheorem{lemma}[theorem]{Lemma}
\newtheorem{corollary}[theorem]{Corollary}
\theoremstyle{definition}

\newtheorem{assumption}[theorem]{Assumption}
\theoremstyle{remark}

\usepackage{algorithm}
\usepackage{algorithmicx}
\usepackage{algpseudocode}

\DeclareMathOperator*{\argmin}{argmin}
\DeclareMathOperator*{\argmax}{argmax}

\usepackage[textsize=tiny]{todonotes}

\author{%
  Claire Chen\\
  The Division of Physics,
Mathematics and Astronomy \\
California Institute of Technology\\
  \texttt{clairechen@caltech.edu} \\
    \And
    Yuheng Zhang \\
  Department of Computer Science \\
  University of Illinois Urbana-Champaign \\
  \texttt{ yuhengz2@illinois.edu} \\
}

\begin{document}

\maketitle
\begin{abstract}
An $\alpha$-potential game is a multi-player non-cooperative interaction in which a global potential function approximates individual player rewards up to a structural bias $\alpha$. While identifying a Nash Equilibrium (NE) in generic general-sum games is known to be computationally intractable, the potential game structure enables tractable NE identification. In this paper, we study the offline learning of NE in $\alpha$-potential games using KL regularization. To analyze this process, we propose a novel \textit{Reference-Anchored} offline data coverage framework that anchors data requirements to a known reference policy rather than an unknown optimum. Building on this, we propose Offline Potential Mirror Descent (OPMD), a decentralized algorithm that achieves an accelerated $\widetilde{\mathcal{O}}(1/n)$ statistical rate, surpassing the standard $\widetilde{\mathcal{O}}(1/\sqrt{n})$ rate typical of offline multi-agent learning. This work characterizes the first fast-rate offline learning approach for $\alpha$-potential games.
\end{abstract}

\section{Introduction}

Finding a Nash equilibrium in general-sum games is PPAD-complete \citep{daskalakis2009complexity, chen2009settling}, establishing that equilibrium learning is computationally intractable in generic general-sum games. This fundamental barrier poses a significant challenge for multi-agent reinforcement learning (MARL), where the goal is to derive equilibrium policies in high-dimensional environments. Fortunately, recent theoretical advancements show that the $\alpha$-potential game framework provides a general form for dynamic $N$-player non-cooperative games \citep{guo2023markov, guo2025potential, li2025alpha}. This structure  reduces the complex task of identifying Nash equilibria to a single-objective optimization problem, as any optimizer of the global potential function is guaranteed to be an equilibrium within a corresponding $\alpha$-error bound \citep{guo2023markov}. Such a reduction provides a mathematically tractable framework that bypasses the complexities of generic general-sum dynamics.

Crucially, it has been shown that the $\alpha$-potential structure is universal. Specifically, any finite non-cooperative game can be characterized as an $\alpha$-potential game, where the magnitude of $\alpha$ represents the inherent structural bias of the system \citep{guo2025markov, li2025alpha}. This universality encompasses practically significant scenarios such as traffic congestion, autonomous driving, smart grid management, and perturbed team games \citep{guo2025markov}. This structural alignment reduces the strategic search for Nash equilibria to a single-objective optimization task, thereby bypassing the PPAD-complete complexity inherent in generic general-sum games \citep{guo2025potential, li2025alpha}.

In the offline setting, the difficulty of learning stable strategies is compounded by distribution shift, where the learned policies explore regions of the state-action space poorly represented in the historical logs \citep{levine2020offline}. To establish sample-efficient recovery, standard methodologies typically rely on the principle of explicit pessimism \citep{jin2021pessimism, zhang2023offline}. These approaches utilize manually designed lower confidence bounds (LCB) or uncertainty bonuses to penalize values in regions with low data support \citep{cui2022provably, zhong2022pessimistic}. While theoretically robust, these pessimistic mechanisms suffer from two primary drawbacks. First, they often require the construction of complex uncertainty quantifiers over joint action spaces, leading to significant computational overhead. Second, existing coverage conditions for these methods are typically anchored to an unknown optimal equilibrium, introducing a conceptual circularity—known as the "Curse of the Unknown Optimum"—that makes the assumptions difficult to verify in practice.

In this work, we establish that KL regularization relative to a fixed and known reference policy provides a standalone, pessimism-free mechanism for equilibrium learning in $\alpha$-potential contextual bandits. To this end, we propose the \textit{Reference-Anchored} coverage framework, a novel formulation of unilateral concentrability that resolves the conceptual circularity of prior work by anchoring data requirements to a reference policy rather than an unknown optimum. We prove that this regularized framework not only stabilizes learning but also leverages the inherent geometry of potential games to achieve superior statistical efficiency. Specifically, we show that decentralized regularized learning can achieve fast $\tilde{\mathcal{O}}(1/n)$ statistical rates, bypassing the standard $\tilde{\mathcal{O}}(1/\sqrt{n})$ bottleneck found in generic multi-agent learning \citep{cui2022provably, zhang2023offline}.

Our analysis centers on a unique geometric discovery: the KL-regularized Mirror Descent update forms an exact log-linear interpolation between the current iterate and the regularized best response. This structural identity allows us to recast the evaluation mismatch between best responses and the learning iterates as a vanishing optimization error; in contrast, traditional analyses are forced to treat this term as a fixed statistical bottleneck that prevents acceleration beyond the standard $\tilde{\mathcal{O}}(1/\sqrt{n})$ rate \citep{cui2022provably, zhong2022pessimistic}. Consequently, we establish that first-order estimation errors cancel out within the decentralized product structure of potential-game iterates,  enabling the accelerated $\tilde{\mathcal{O}}(1/n)$ rate. Our key contributions are summarized as follows:

\begin{itemize}
    \item \textbf{Reference-Anchored Coverage}: We introduce a novel unilateral concentrability condition that anchors data requirements to a known reference policy rather than an unknown optimum. This formulation provides a tractable framework for offline MARL, resolving the conceptual circularity of prior work.
    
    \item \textbf{Theoretical Fast-Rate Framework}: We propose the Regularized Offline Potential Equilibrium (ROPE) algorithm to establish that a fast $\tilde{\mathcal{O}}(1/n)$ statistical rate is achievable in regularized $\alpha$-potential games. Our analysis demonstrates that the structural alignment of the potential function enables the cancellation of first-order errors, bypassing the $\tilde{\mathcal{O}}(1/\sqrt{n})$ bottleneck typical of generic general-sum games.
    
    \item \textbf{Decentralized Algorithmic Recovery}: We develop Offline Potential Mirror Descent (OPMD), a computationally tractable algorithm that implements the regularized equilibrium learning process through iterative self-play. We prove that OPMD preserves the $\tilde{\mathcal{O}}(1/n)$ rate (up to the structural bias $\alpha$) by exploiting the log-linear geometry of the Mirror Descent updates to resolve linear evaluation mismatch as a vanishing optimization error.
\end{itemize}



\section{Related Work}
\label{sec:related_work}

\paragraph{Potential Games.}
The study of dynamic potential games generalizes the static framework to state-dependent environments, ensuring the guaranteed existence of pure Nash Equilibria and the global convergence of decentralized gradient play \citep{leonardos2021global, yan2025markov}. To accommodate the diverse structures of real-world strategic interactions, recent theoretical advancements have introduced the $\alpha$-potential framework, which characterizes systems where unilateral deviations track potential changes up to a structural bias $\alpha \ge 0$ \citep{guo2025potential, li2025alpha}. Within this landscape, \citet{guo2025potential} established that the $\alpha$-potential structure is universal, proving that any finite-state dynamic game can be characterized as an $\alpha$-potential game. This characterization has been extended to stochastic differential games and constrained settings, with explicit characterizations of $\alpha$ in terms of player heterogeneity and interaction intensity \citep{guo2025potential, das2024learning}. Furthermore, \citet{yan2025markov} identified sufficient conditions on reward design and transition dynamics that allow for the systematic construction of potential games, with specific applications to autonomous driving and traffic management.

While the $\alpha$-potential framework reduces equilibrium search to a single-objective optimization problem, existing literature primarily focuses on the \textit{planning} regime—where the optimizer has access to the potential landscape or best-response oracles \citep{guo2023markov, das2024learning}—or the \textit{online learning} regime. In the latter, \citet{ding2022independent} proposed independent policy gradient (IPG) algorithms for large-scale games, establishing iteration and sample complexity bounds that remain independent of the state-space size. Similarly, \citet{li2025alpha} provided a comprehensive analysis of RL algorithms in $\alpha$-potential games across both discrete and continuous time. However, these results are derived for the online setting, where agents benefit from active exploration and immediate environmental feedback. Despite the computational tractability afforded by the potential structure, the statistical properties of equilibrium learning in the \textit{offline} regime—specifically under general function approximation and distribution shift—remain largely unexplored. Our work bridges this gap by demonstrating that the structural alignment of $\alpha$-potential games can be exploited to achieve accelerated statistical rates in the offline setting without the need for explicit pessimistic mechanisms.

\paragraph{General-Sum Games.}
The transition from single-agent reinforcement learning to multi-agent environments necessitates a shift from identifying value-optimal policies to recovering stable equilibria \citep{roughgarden2016twenty, cui2022offline}. While single-agent objectives seek to maximize a unified return, multi-agent dynamics are characterized by individual strategic incentives, typically modeled through Coarse Correlated Equilibria (CCE) or Nash Equilibria (NE) \citep{daskalakis2009complexity, cui2022provably}. Recovering a Nash Equilibrium is fundamentally challenging due to its PPAD-complete complexity in general-sum settings, which has historically led many methodologies to focus on the more computationally tractable CCE, where stability is established via no-regret learning or correlated strategies \citep{chen2009settling, yan2024model}. However, the pursuit of Nash stability remains essential for decentralized applications requiring robust independent decision-making, where coordinated or correlated mechanisms are infeasible or strategically undesirable \citep{song2021can, cui2022provably}.

In the offline regime, the difficulty of equilibrium learning is compounded by the lack of interactive exploration and the presence of severe distribution shift \citep{levine2020offline, jin2021pessimism}. Standard methodologies typically address these gaps by adapting the principle of pessimism to the game-theoretic setting, utilizing uncertainty bonuses or lower confidence bounds (LCB) to penalize values in poorly explored joint-action regions \citep{zhong2022pessimistic, cui2022provably, zhang2023offline}. Despite their theoretical robustness, these pessimistic frameworks suffer from two primary bottlenecks. First, many existing approaches rely on complex uncertainty quantifiers over joint action spaces, leading to significant computational overhead unless restricted by specific unilateral concentrability conditions \citep{cui2022provably, yan2024model}. Second, current algorithms that target Nash Equilibria often require centralized coordination or heavy iterative value updates that lack decentralized implementations capable of achieving accelerated statistical rates \citep{zhong2022pessimistic, zhang2023offline}. Our work addresses these challenges by demonstrating that regularized decentralized learning can successfully recover Nash Equilibria at fast statistical rates without the need for explicit pessimistic bonuses, bypassing the computational and complexity barriers characteristic of traditional general-sum learning.

\section{Preliminaries}
\label{sec:preliminaries}

We consider an $m$-player contextual bandit game possessing an $\alpha$-potential structure. In this section, we formalize the environment, the regularized objectives, and the potential-based equilibrium concepts analyzed in this work.

\subsection{$m$-Player Contextual Bandits}
The game is defined by the tuple $\mathcal{M} = (\mathcal{N}, \mathcal{X}, \{\mathcal{A}_i\}_{i=1}^m, \{r_i\}_{i=1}^m, \rho)$, where $\mathcal{N} = [m]$ is the set of players, $\mathcal{X}$ is the context space, and $\mathcal{A}_i$ is the finite action space for player $i$. In each round, a context $x \sim \rho$ is sampled. Players simultaneously select actions $a_i \in \mathcal{A}_i$, forming a joint action profile $\boldsymbol{a} = (a_1, \dots, a_m) \in \mathcal{A} \coloneqq \prod_{i=1}^m \mathcal{A}_i$. Each player $i$ receives a reward $r_i(x, \boldsymbol{a}) \in [0, 1]$. A joint policy $\pi: \mathcal{X} \to \Delta(\mathcal{A})$ maps contexts to distributions over joint actions. In this work, we primarily focus on product policies, where $\pi(\boldsymbol{a} \mid x) = \prod_{i=1}^m \pi_i(a_i \mid x)$.

\subsection{KL-Regularized Objectives}

KL regularization relative to a fixed reference policy has emerged as a fundamental mechanism for enforcing behavioral constraints and incorporating prior knowledge in multi-agent systems. In many practical scenarios, such as the alignment of large language models through Reinforcement Learning from Human Feedback, KL regularization is a functional requirement to prevent policy drift and ensure the model remains anchored to a trusted reference distribution \citep{munos2024nash}. Within the offline setting, this regularization provides a dual benefit: it acts as a stabilizer against distribution shift and serves as an implicit pessimism mechanism that replaces the need for explicit, manually tuned pessimistic bonuses.

For a fixed reference joint policy $\pi^{\mathrm{ref}} = (\pi_1^{\mathrm{ref}}, \dots, \pi_m^{\mathrm{ref}})$ and a regularization coefficient $\eta > 0$, the regularized value function $V_i^\pi$ for player $i$ is defined as:
\begin{equation}
    V_i^{\pi}(x) \coloneqq \mathbb{E}_{\boldsymbol{a} \sim \pi(\cdot \mid x)} [r_i(x, \boldsymbol{a})] - \eta^{-1} \mathrm{KL}\big(\pi_i(\cdot \mid x) \,\|\, \pi_i^{\mathrm{ref}}(\cdot \mid x)\big).
    \label{eq:value_kl_def}
\end{equation}
 We denote the global regularized return as $J_i(\pi) \coloneqq \mathbb{E}_{x \sim \rho} [V_i^\pi(x)]$. Each player $i$ seeks to maximize their own regularized return given the strategies of their opponents $\pi_{-i}$. The corresponding regularized best-response value is defined as $V_i^{\dagger, \pi_{-i}}(x) \coloneqq \max_{\pi_i'} V_i^{\pi_i', \pi_{-i}}(x)$. 

\subsection{The $\alpha$-Potential Structure}
We leverage the framework of $\alpha$-potential games to characterize the strategic interactions \citep{guo2025markov, li2025alpha}. A game is an $\alpha$-potential game if there exists a global potential function $\Phi: \Pi \to \mathbb{R}$ such that for any player $i$, any joint policy $\pi$, and any unilateral deviation $\pi_i'$, the change in the individual return is tracked by the change in the potential function up to a structural bias $\alpha \ge 0$:
\begin{equation}
    \left| \left( J_i(\pi_i', \pi_{-i}) - J_i(\pi_i, \pi_{-i}) \right) - \left( \Phi(\pi_i', \pi_{-i}) - \Phi(\pi_i, \pi_{-i}) \right) \right| \le \alpha.
\end{equation}
This framework is universal, as any finite non-cooperative game can be characterized as an $\alpha$-potential game for a sufficiently large $\alpha$ \citep{li2025alpha, guo2025potential}. In this work, we exploit this structure to reduce the search for equilibria to a potential maximization problem.

\subsection{$\alpha$-Nash Equilibrium and Total Exploitability}

Our goal is to recover a joint policy $\pi^*$ that constitutes an $\alpha$-Nash Equilibrium ($\alpha$-NE). Conceptually, an $\alpha$-NE represents a state of near-stability where no player can unilaterally improve their regularized return by more than the structural bias $\alpha$:
\begin{equation}
    J_i(\pi_i^*, \pi_{-i}^*) \ge \sup_{\pi_i'} J_i(\pi_i', \pi_{-i}^*) - \alpha, \quad \forall i \in [m]. 
\end{equation}
For a learned policy $\hat{\pi}$, we evaluate its performance using the NE Total Exploitability, or Nash Gap, which aggregates individual suboptimality to provide a scalar measure of the policy's distance from equilibrium:
\begin{equation}
    \mathrm{Gap}_{\mathrm{NE}}(\hat{\pi}) \coloneqq \sum_{i=1}^m \mathbb{E}_{x \sim \rho} \left[ V_i^{\dagger, \hat{\pi}_{-i}}(x) - V_i^{\hat{\pi}}(x) \right]. 
\end{equation}
In this formulation, $V_i^{\dagger, \hat{\pi}_{-i}}(x) - V_i^{\hat{\pi}}(x)$ represents the maximum unilateral improvement available to player $i$ given the fixed strategies of their opponents.

The $\alpha$-potential structure facilitates a reduction of this equilibrium search to a single-objective optimization problem \citep{guo2023markov}. Formally, let $\hat{\pi}$ denote an $\epsilon$-approximate maximizer of the global potential function $\Phi$, satisfying:
\begin{equation}
    \Phi(\hat{\pi}) \ge \sup_{\pi \in \Pi} \Phi(\pi) - \epsilon 
\end{equation}
for some optimization error $\epsilon \ge 0$. A core result in the characterization of $\alpha$-potential games is that any such policy $\hat{\pi}$ is guaranteed to be an $(\epsilon + 2\alpha)$-Nash Equilibrium of the underlying regularized game \citep{guo2025markov}. This structural alignment establishes potential optimization as a computationally tractable proxy for equilibrium learning, thereby bypassing the PPAD-complete complexity characteristic of generic general-sum interactions.

\subsection{Offline Learning Model}
The learner has access to an offline dataset $\mathcal{D} = \{ (x_{\tau}, \boldsymbol{a}_{\tau}, \mathbf{r}_{\tau}) \}_{\tau=1}^{n}$ generated by an unknown behavioral distribution $\mu(x, \boldsymbol{a})$ with $\mathbf{r}_{\tau} = (r_{\tau,1}, \dots, r_{\tau,m})$. We assume access to function classes $\{\mathcal{Q}_i\}_{i=1}^m$ for reward estimation, where each class $\mathcal{Q}_i$ consists of functions mapping $\mathcal{X} \times \mathcal{A} \to [0, 1]$. Following standard literature \citep{xie2021batch, zhang2023offline, liu2025ode}, we assume realizability, i.e., $r_i \in \mathcal{Q}_i$. For any estimate $\hat{Q}_i \in \mathcal{Q}_i$, the pointwise regression error is defined as $\mathcal{Z}_i(x, \boldsymbol{a}) \coloneqq \hat{Q}_i(x, \boldsymbol{a}) - r_i(x, \boldsymbol{a})$.


\section{Data Coverage and Concentrability}
\label{sec:coverage}

A fundamental challenge in offline reinforcement learning is the distribution shift between the logged dataset $\mathcal{D}$ and the target policy \citep{levine2020offline, cui2022provably, zhang2023offline}. To establish sample-efficient recovery, one must assume the offline data provides sufficient coverage of the state-action regions relevant to the target objective. In this section, we review the landscape of concentrability assumptions in multi-agent systems and formalize our novel reference-anchored approach.
\subsection{From Single-Agent to Joint Concentrability}

In single-agent offline RL, sample efficiency is typically established under the assumption that the data distribution $\mu$ covers the optimal policy $\pi^*$ \citep{levine2020offline}. 

\begin{assumption}[Single-Agent Concentrability]
There exists a constant $C^* \ge 1$ such that 
\begin{align}
    \left\| \frac{\rho(x)\pi^*(a \mid x)}{\mu(x, a)} \right\|_\infty \le C^*.
\end{align}
\end{assumption}

However, in multi-agent games, simply covering a single optimal profile is insufficient for learning a stable equilibrium \citep{cui2022offline}. The most straightforward extension is \textbf{joint concentrability}, which demands that the dataset support all possible simultaneous deviations by the agents.

\begin{assumption}[Joint Multi-Agent Concentrability]
There exists a constant $C_{\mathrm{joint}} \ge 1$ such that for any joint policy $\pi = \prod_{i=1}^m \pi_i$:
\begin{align}
    \left\| \frac{\rho(x)\pi(\boldsymbol{a} \mid x)}{\mu(x, \boldsymbol{a})} \right\|_\infty \le C_{\mathrm{joint}}.
\end{align}
\end{assumption}

Here, the vector $\boldsymbol{a}$ denotes the joint action profile of all $m$ players. This requirement scales with the \textit{product }of action spaces ($\prod_{i=1}^m |\mathcal{A}_i|$), resulting in an \textit{exponential }growth in complexity known as the ``curse of multiagents'' \citep{cui2022provably}. This scaling makes the assumption practically prohibitive for environments with many participants or a large action space. In contrast, \citet{cui2022provably} demonstrate that a strategy-wise approach can achieve linear scaling ($\sum_{i=1}^m |\mathcal{A}_i|$), effectively breaking this curse by reducing the complexity to the sum of individual action spaces.

\subsection{Unilateral Concentrability and the Circularity Problem}

To bypass the exponential scaling of joint action spaces, recent works have adopted \textbf{unilateral concentrability} \citep{cui2022provably, cui2022offline, zhang2023offline}. This principle requires that the dataset only needs to support scenarios where a single agent deviates from a target equilibrium strategy $\pi^*$ while all others remain stationary.

\begin{assumption}[Equilibrium-Anchored Unilateral Concentrability, \citet{cui2022provably}]
Let $\pi^* = (\pi_1^*, \dots, \pi_m^*)$ be a target equilibrium policy. Let $\Pi_{\mathrm{uni}}(\pi^*) \coloneqq \bigcup_{i=1}^m \left( \Pi_i \times \{\pi_{-i}^*\} \right)$ be the set of unilateral deviations from $\pi^*$. There exists a constant $C_{\mathrm{uni}}^* \ge 1$ such that for any $\pi' \in \Pi_{\mathrm{uni}}(\pi^*)$:
\begin{align}
    \left\| \frac{\rho(x)\pi'(\boldsymbol{a} \mid x)}{\mu(x, \boldsymbol{a})} \right\|_\infty \le C_{\mathrm{uni}}^*.
\end{align}
\end{assumption}

\citet{cui2022provably} demonstrated that unilateral concentration is a necessary assumption for learning a Nash Equilibrium (NE) in offline games. Crucially, \citet{cui2022provably} demonstrate that this strategy-wise approach allows the sample complexity to scale \textit{linearly }with the \textit{sum }of action spaces ($\sum_{i=1}^m |\mathcal{A}_i|$), effectively breaking the curse of multiagents.

\paragraph{The Curse of the Unknown Optimum.} Despite its improved scaling, classical unilateral concentrability is anchored to an unknown optimal equilibrium policy $\pi^*$ \citep{cui2022offline, zhang2023offline}. This introduces a fundamental conceptual circularity: to define the data sufficiency required to find the equilibrium, one must already know the equilibrium itself. In general-sum settings, where equilibria are potentially non-unique and finding one is PPAD-complete \citep{daskalakis2009complexity}, this ambiguity makes the assumption theoretically ill-posed and practically unverifiable.

\subsection{Proposed: Reference-Anchored Unilateral Concentrability}

By utilizing KL regularization, we solve the circularity problem by anchoring the coverage requirement to a \textit{fixed and known} reference policy $\pi^{\mathrm{ref}}$. This shifts the theoretical burden from an elusive optimal strategy to a concrete anchor.

\begin{assumption}[Reference-Anchored Unilateral Concentrability]
\label{ass:concentrability}
Let $\Pi = \prod_{i=1}^m \Pi_i$ be the class of joint policies, and let $\pi^{\mathrm{ref}} = (\pi_1^{\mathrm{ref}}, \dots, \pi_m^{\mathrm{ref}})$ be the fixed reference policy. We define the set of \emph{reference-anchored unilateral deviation policies} $\Pi_{\mathrm{ref-uni}} \subset \Pi$ as the union of policies where exactly one player deviates to an arbitrary policy while all other players adhere to the reference policy: 
\begin{align}
    \Pi_{\mathrm{ref-uni}} \coloneqq \bigcup_{i=1}^m \left( \Pi_i \times \{\pi_{-i}^{\mathrm{ref}}\} \right).
\end{align}
Let $\mu \in \Delta(\mathcal{X} \times \mathcal{A})$ denote the underlying joint context-action distribution of the offline dataset $\mathcal{D}$. We assume there exists a constant $C_{\mathrm{uni}} \ge 1$ such that for any $\pi' \in \Pi_{\mathrm{ref-uni}}$, the density ratio is uniformly bounded: 
\begin{align}
    \left\| \frac{\rho(x)\pi'(\boldsymbol{a} \mid x)}{\mu(x, \boldsymbol{a})} \right\|_\infty \le C_{\mathrm{uni}}.
\end{align}
\end{assumption}

\paragraph{Advantages of Reference-Anchored Coverage.} 
Assumption~\ref{ass:concentrability} offers two critical advantages over prior formulations:
\begin{enumerate}
    \item \textbf{Verifiability}: Unlike $\pi^*$, the reference policy $\pi^{\mathrm{ref}}$ is known to the learner. It can be chosen as the behavior policy that generated the data or a previously validated iteration, allowing the coverage coefficient $C_{\mathrm{uni}}$ to be empirically estimated.
    \item \textbf{Tractability}: By focusing on unilateral deviations, we maintain the $\sum_{i=1}^m |\mathcal{A}_i|$ scaling derived in \citet{cui2022provably}, effectively avoiding the exponential joint action space while providing a more robust stabilizer through KL anchoring.
\end{enumerate}

We further denote $C_{\mathrm{shift}} \coloneqq \exp(\eta(m-1))$ to quantify the density ratio shift between a learned policy and the reference policy across $m-1$ opponents. Combined with $C_{\mathrm{uni}}$, it characterizes the total statistical overhead for reference-anchored offline equilibrium learning.

\section{Regularized Offline Potential Equilibrium}
\label{sec:rope_algorithm}

In this section, we establish the Regularized Offline Potential Equilibrium (ROPE) framework, which identifies $\alpha$-Nash Equilibria in $\alpha$-potential contextual bandits by solving a regularized empirical game. By incorporating KL regularization relative to a known reference policy, ROPE provides a mechanism for stable equilibrium recovery from offline data without relying on the manually tuned, explicit pessimistic bonuses often required in standard methodologies \citep{jin2021pessimism, cui2022provably}.

The fundamental premise of the ROPE framework is to leverage the structural alignment between individual player utilities and the global potential function \citep{li2025alpha, guo2025markov}. In $\alpha$-potential games, the search for a strategic equilibrium is computationally equivalent to a potential maximization task, where any approximate maximizer of the potential corresponds to an $\alpha$-Nash Equilibrium \citep{guo2025markov}. 

\subsection{The ROPE Framework}

The ROPE algorithm (Algorithm~\ref{alg:ideal_rope}) operates through a two-stage process: statistical estimation followed by strategic computation. First, the learner utilizes the offline dataset $\mathcal{D}$ to construct empirical estimates of the reward functions for each player via least-squares regression. These estimates serve as the foundation for the empirical game. Second, the algorithm computes a joint policy that constitutes an exact regularized Nash Equilibrium for this estimated environment. 

\begin{algorithm}[H]
\caption{Regularized Offline Potential Equilibrium (ROPE)}
\label{alg:ideal_rope}
\begin{algorithmic}[1]
    \State {\bfseries Input:} Offline dataset $\mathcal{D}$, reference policies $\pi^{\mathrm{ref}}$, regularization parameter $\eta > 0$, function classes $\{\mathcal{Q}_i\}_{i=1}^m$
    \For{player $i = 1, \dots, m$}
        \State {\bfseries Empirical Estimation:} Estimate the Q-function via least-squares regression:
        \State $\quad \hat{Q}_i \leftarrow \argmin_{f \in \mathcal{Q}_i} \sum_{\tau=1}^n \left( f(x_\tau, \boldsymbol{a}_\tau) - r_{\tau,i} \right)^2$
    \EndFor
    
    \State {\bfseries Equilibrium Computation:} For each context $x \in \mathcal{X}$, compute a joint product policy 
    $\hat{\pi}(\cdot|x) = \prod_{i=1}^m \hat{\pi}_i(\cdot|x)$ 
    satisfying the exact regularized best-response condition for all $i \in [m]$:
    \State $\quad \hat{\pi}_i(\cdot|x) \in \argmax_{\pi_i \in \Delta(\mathcal{A}_i)} \left\{ \mathbb{E}_{a_i \sim \pi_i, \boldsymbol{a}_{-i} \sim \hat{\pi}_{-i}} [ \hat{Q}_i(x, \boldsymbol{a}) ] - \eta^{-1} \mathrm{KL}(\pi_i \| \pi_i^{\mathrm{ref}}) \right\}$
    
    \For{player $i = 1, \dots, m$}
        \State {\bfseries Value Update:} Update the empirical value function for Player $i$:
        \State $\quad \hat{V}_i(x) \leftarrow \mathbb{E}_{\boldsymbol{a} \sim \hat{\pi}}[ \hat{Q}_i(x, \boldsymbol{a}) ] - \eta^{-1}\mathrm{KL}(\hat{\pi}_i \| \pi_i^{\mathrm{ref}})$
    \EndFor
    \State {\bfseries Output:} Learned joint policy $\hat{\pi}$
\end{algorithmic}
\end{algorithm}

\paragraph{Algorithmic Breakdown.} 
The estimation phase in ROPE (Lines 2–5) relies on standard empirical risk minimization, where the complexity of the function classes $\{\mathcal{Q}_i\}_{i=1}^m$ and the dataset size $n$ govern the precision of the reward approximations \citep{cui2022provably, zhang2023offline}. The equilibrium computation phase (Lines 6–7) requires identifying a joint product policy $\hat{\pi}$ where every player's strategy is a best response to the profile of their opponents within the empirical game. 

Crucially, the $\alpha$-potential structure ensures that this fixed-point condition aligns with the global maximization of the latent potential function $\Phi$ \citep{li2025alpha}. In generic general-sum games, identifying such an equilibrium is generally PPAD-complete \citep{chen2009settling}. However, the structural symmetry of the potential game, combined with the smoothing effect of the KL regularizer, ensures that this objective is well-posed and serves as a reliable proxy for recovery in the true environment.

\subsection{Statistical Guarantees}

By leveraging the Reference-Anchored Unilateral Concentrability framework (Assumption~\ref{ass:concentrability}), we establish that the equilibrium identified by ROPE achieves superior statistical efficiency. The following theorem demonstrates that anchoring to a known reference policy allows the learner to exceed the standard statistical rates of generic multi-agent systems.

\begin{theorem}
\label{thm:fast_rate_rope}
Under Assumption~\ref{ass:concentrability}, and assuming the reward functions are realizable within $\{\mathcal{Q}_i\}_{i=1}^m$, the estimated regularized joint policy $\hat{\pi}$ returned by ROPE achieves, with probability at least $1-\delta$, a Total Exploitability bounded by:
\begin{align}
    \mathrm{Gap}_{\mathrm{NE}}(\hat{\pi}) \le \widetilde{\mathcal{O}}\left( \frac{m \eta C_{\mathrm{shift}} C_{\mathrm{uni}} \log (|\mathcal{Q}|/\delta)}{n} \right).
\end{align}
\end{theorem}

The proof of Theorem~\ref{thm:fast_rate_rope} is provided in Appendix~\ref{app:close-rope}. The accelerated $\widetilde{\mathcal{O}}(1/n)$ statistical rate is a consequence of the structural alignment between the learned equilibrium $\hat{\pi}$ and the evaluation trajectories of the best-response values. Specifically, within the per-iterate decomposition $V_i^{\dagger,\hat{\pi}_{-i}} - V_i^{\hat{\pi}} = \varepsilon_{\mathrm{pot},i} + \delta_{\mathrm{BR},i} + \delta_{\mathrm{Iter},i}$ (Equation~\eqref{eq:gap_decomp_named} in Appendix~\ref{app:proofs}), the linear parts of the evaluation bias $\delta_{\mathrm{BR},i}$ and the iterate noise $\delta_{\mathrm{Iter},i}$ expand into expectations of the regression error $\mathcal{Z}_i$ with opposite signs. Because $\hat{\pi}$ is a product policy, these evaluation distributions match perfectly, the linear $\mathcal{Z}_i$ terms cancel exactly, and the potential ascent error $\varepsilon_{\mathrm{pot},i}$ vanishes by exact-NE optimality. By the Bregman bound on the regularized best-response gap (Lemma~\ref{lem:bregman_bestresponse}, derived from Pinsker's inequality and Fenchel duality), the remaining exploitability gap is dominated by a second-order squared residual. Under the unilateral data coverage guaranteed by Assumption~\ref{ass:concentrability}, this residual scales with the in-sample squared regression error, which decays at the fast $\widetilde{\mathcal{O}}(1/n)$ rate. This result demonstrates that for environments with potential-aligned structures, solving the regularized empirical game is sufficient to bypass the standard $\widetilde{\mathcal{O}}(1/\sqrt{n})$ statistical bottleneck characteristic of generic offline learning.

\section{Offline Potential Mirror Descent}
\label{sec:practical_algorithm}

While the ROPE framework (Algorithm~\ref{alg:ideal_rope}) establishes that maximizing the latent potential function yields an $\alpha$-Nash Equilibrium, it assumes access to an exact optimization oracle. In high-dimensional environments, explicitly constructing and optimizing the global potential $\Phi$ is often computationally prohibitive. To bridge this gap, we leverage the fact that in $\alpha$-potential games, individual improvements naturally translate into global progress \citep{guo2023markov,guo2025markov}. The proposed Algorithm~\ref{alg:sos_md} (OPMD) exploits this structure; it allows players to maximize their local utilities independently, serving as a practical, decentralized approach for maximizing the global potential.


\subsection{The OPMD Algorithm}

The OPMD algorithm utilizes decentralized self-play to recover strategic equilibria without the need for a global solver. The process begins by estimating individual Q-functions via least-squares regression over the offline dataset $\mathcal{D}$. Following estimation, players enter an iterative inner loop where they perform local KL-regularized Mirror Descent updates based on the marginalized actions of their opponents. 


Upon completion, the algorithm returns an iterate $\tilde{\pi} = \pi^{(t^*)}$ from the generated sequence. Because the game possesses an $\alpha$-potential structure, the gradient of the global potential function $\Phi$ is composed of the individual player gradients. Consequently, when players independently update their strategies via Mirror Descent, the joint policy trajectory effectively performs decentralized ascent on the potential landscape. This structural alignment allows the algorithm to resolve the linear evaluation mismatch in $\delta_{\mathrm{BR},i}^{(t)}+\delta_{\mathrm{Iter},i}^{(t)}$ between iterates and best responses as a vanishing optimization error, rather than a fixed statistical bottleneck, attaining the accelerated $\widetilde{\mathcal{O}}(1/n)$ statistical rate characterized in Theorem~\ref{thm:opmd_rate}.


\begin{algorithm}[H]
\caption{Offline Potential Mirror Descent (OPMD)}
\label{alg:sos_md}
\begin{algorithmic}[1]
    \State {\bfseries Input:} Offline dataset $\mathcal{D}$, reference policies $\pi^{\mathrm{ref}}$, regularization $\eta > 0$, function classes $\{\mathcal{Q}_i\}_{i=1}^m$, step size $\gamma > 0$, iterations $T$
    \For{player $i = 1, \dots, m$}
        \State {\bfseries Empirical Estimation:} Estimate Q-functions via least-squares:
        \State $\quad \hat{Q}_i \leftarrow \argmin_{f \in \mathcal{Q}_i} \sum_{\tau=1}^n \left( f(x_\tau, \boldsymbol{a}_\tau) - r_{\tau,i} \right)^2$
        \State Initialize base policy: $\pi_i^{(1)}(\cdot|x) = \pi_i^{\mathrm{ref}}(\cdot|x)$ for all $x \in \mathcal{X}$
    \EndFor
    
    \State {\bfseries Decentralized Self-Play:}
    \For{$t = 1, \dots, T$}
        \For{player $i = 1, \dots, m$}
            \State {\bfseries Marginalized Value:} Evaluate expected Q-value against opponents' current policies:
            \State $\quad \bar{Q}_i^{(t)}(x, a_i) = \mathbb{E}_{\boldsymbol{a}_{-i} \sim \pi_{-i}^{(t)}(\cdot|x)} \big[ \hat{Q}_i(x, a_i, \boldsymbol{a}_{-i}) \big]$
            \State {\bfseries Mirror Descent Update:} Perform local KL-regularized policy improvement:
            \State $\quad \pi_i^{(t+1)}(a_i|x) \propto \left( \pi_i^{\mathrm{ref}}(a_i|x) \right)^{\frac{\gamma}{\eta + \gamma}} \left( \pi_i^{(t)}(a_i|x) \right)^{\frac{\eta}{\eta + \gamma}} \exp\left( \frac{\eta \gamma}{\eta + \gamma} \bar{Q}_i^{(t)}(x, a_i) \right)$
        \EndFor
    \EndFor
    
    \State {\bfseries Output Selection:} Sample an index $t^* \sim \text{Uniform}(1, \dots, T)$
    \State {\bfseries Output:} The product joint policy $\tilde{\pi} = \pi^{(t^*)}$
\end{algorithmic}
\end{algorithm}

\subsection{Statistical Guarantees}

By exploiting the inherent geometry of the Mirror Descent iterates, OPMD (Algorithm~\ref{alg:sos_md}) achieves the fast statistical rate established for the ROPE framework, subject to the structural bias of the game. The expected suboptimality decomposes via $V_i^{\dagger,\pi_{-i}^{(t)}}-V_i^{\pi^{(t)}} = \varepsilon_{\mathrm{pot},i}^{(t)}+\delta_{\mathrm{BR},i}^{(t)}+\delta_{\mathrm{Iter},i}^{(t)}$ (potential-ascent residual, evaluation bias, and iterate noise; see Equation~\eqref{eq:gap_decomp_named}, Appendix~\ref{app:proofs}).

\begin{theorem}
\label{thm:opmd_rate}
Under Assumption~\ref{ass:concentrability}, and assuming reward realizability within $\{\mathcal{Q}_i\}_{i=1}^m$, the policy $\tilde{\pi}$ returned by OPMD satisfies, with probability at least $1-\delta$:
\begin{align}
    \mathbb{E}_{t^*}[\mathrm{Gap}_{\mathrm{NE}}(\tilde{\pi})]
    \;\le\; \underbrace{\widetilde{\mathcal{O}}\!\left(\tfrac{m\sqrt{m}}{\sqrt{T}}\right)}_{\text{optimization}}
    + \underbrace{m\alpha}_{\text{bias}}
    + \underbrace{\widetilde{\mathcal{O}}\!\left(\tfrac{m\eta\,C_{\mathrm{shift}}\,C_{\mathrm{uni}}\,\log(|\mathcal{Q}|/\delta)}{n}\right)}_{\text{statistical}}.
\end{align}
Consequently, for $T \ge n^2$, OPMD recovers an $\alpha$-NE at the fast $\widetilde{\mathcal{O}}(1/n)$ statistical rate.
\end{theorem}

\subsection{Proof Sketch of Theorem~\ref{thm:opmd_rate}}

The proof establishes that decentralized regularized updates successfully resolve the statistical mismatch bottleneck. We outline the core mechanism below; the complete derivation is provided in Appendix~\ref{app:proofs}.

\paragraph{Per-Iterate Decomposition and Potential Ascent.} We decompose the per-iterate gap into potential-ascent residual $\varepsilon_{\mathrm{pot},i}^{(t)}$, evaluation bias $\delta_{\mathrm{BR},i}^{(t)}$, and iterate noise $\delta_{\mathrm{Iter},i}^{(t)}$ \eqref{eq:gap_decomp_named}. By Lemma~\ref{lem:potential_ascent}, decentralized Mirror Descent updates monotonically increase the regularized potential $\hat{\Phi}$, causing $\sum_i \varepsilon_{\mathrm{pot},i}^{(t)}$ to decay at $\mathcal{O}(1/\sqrt{T})$ toward the $m\alpha$ floor.

\paragraph{Log-Linear Interpolation and Mismatch Resolution.} Crucially, the update forms an exact log-linear interpolation between the current iterate and its regularized best response:
\begin{align}
    \pi_i^{(t+1)} \propto \left( \pi_i^{(t)} \right)^{\frac{\eta}{\eta+\gamma}} \left( \pi_i^{\dagger, (t)} \right)^{\frac{\gamma}{\eta+\gamma}}. \label{eq:loglin_sketch}
\end{align}
Per Lemma~\ref{lem:l1_proportionality}, this identity ensures the $L_1$ distance to the optimum is proportional to the step distance $\|\pi_i^{(t)} - \pi_i^{(t+1)}\|_1$. Summing these steps over the iteration sequence recasts the linear evaluation mismatch as a vanishing $\mathcal{O}(1/\sqrt{T})$ optimization error.

\paragraph{Final Synthesis.} The remaining $\delta_{\mathrm{BR},i}^{(t)}+\delta_{\mathrm{Iter},i}^{(t)}$ contribution is dominated by second-order regression terms which decay at the fast $\mathcal{O}(1/n)$ rate under the unilateral coverage of Assumption~\ref{ass:concentrability}. Setting $T \ge n^2$ ensures that computational optimization errors are dominated by the statistical limit.

\section{Conclusion}
\label{sec:conclusion}

This paper established that $\alpha$-potential games allow for accelerated statistical efficiency in offline equilibrium learning through KL regularization. We introduced the ROPE framework (Algorithm~\ref{alg:ideal_rope}) and the decentralized OPMD algorithm (Algorithm~\ref{alg:sos_md}), demonstrating that regularized Mirror Descent achieves a fast $\widetilde{\mathcal{O}}(1/n)$ statistical rate for $\alpha$-Nash Equilibria (Theorem~\ref{thm:fast_rate_rope} and Theorem~\ref{thm:opmd_rate}). By exploiting the geometric structure of potential-aligned product policies, our approach resolves the evaluation mismatch ($\delta_{\mathrm{BR},i}^{(t)}+\delta_{\mathrm{Iter},i}^{(t)}$ in our per-iterate decomposition) between iterates and best responses as a vanishing optimization error. This allows the algorithm to bypass the standard $\widetilde{\mathcal{O}}(1/\sqrt{n})$ statistical bottleneck characteristic of generic offline learning without requiring explicit pessimistic bonuses.

Future research may extend these fast-rate guarantees to multi-step environments and explore the adversarial robustness of recovered equilibria under perturbations to the transition dynamics. Additionally, investigating alternative divergences beyond the KL penalty may reveal broader classes of geometric identities for resolving distribution shift in complex strategic interactions. By establishing a pessimism-free foundation for fast-rate recovery, our work provides a statistically efficient path toward large-scale strategic decision-making in complex multi-agent environments.

\section*{Acknowledgment}
Yuheng Zhang is supported by a fellowship from the Amazon-Illinois Center on AI for Interactive Conversational Experiences (AICE). Nan Jiang acknowledges funding support from NSF CNS-2112471, NSF CAREER IIS-2141781, and Sloan Fellowship.

\bibliographystyle{apalike}
\bibliography{bibliography}

\clearpage
\appendix
\section{Additional Related Work}
\paragraph{Foundations of Offline Reinforcement Learning.}
Offline reinforcement learning (RL) addresses the challenge of deriving optimal policies from static datasets without environment interaction \citep{lange2012batch, levine2020offline}. The central difficulty in this regime is the presence of distribution shift, where the target policy's visitation frequency diverges from the logged behavior, often leading to divergent value estimates and catastrophic extrapolation errors \citep{fujimoto2019off, kumar2019stabilizing, liu2024doubly, liu2024efficient, chen2025efficient, liu2025efficient}. To establish sample-efficient guarantees, the literature has converged on the principle of \textit{pessimism}, which penalizes value estimates in state-action regions lacking sufficient data support \citep{liu2020provably, xie2021bellman, uehara2021pessimistic}. These conservative frameworks, often implemented via lower confidence bounds (LCB) or conservative value iteration, have been shown to achieve minimax optimality across a variety of linear and general function approximation settings \citep{rashidinejad2021bridging, jin2021pessimism, li2024settling}.

\paragraph{Complexity and Learning in Multi-Agent Games.}The transition from single-agent RL to multi-agent strategic environments significantly expands the computational and statistical complexity of the learning task. Finding a Nash Equilibrium (NE) in generic general-sum games is established as PPAD-complete, suggesting that equilibrium identification is fundamentally intractable for broad classes of interactions \citep{papadimitriou1994complexity, chen2009settling, daskalakis2009complexity}. In the offline setting, this complexity is further compounded by the "curse of multiagents," where naive joint-action coverage requirements scale exponentially with the number of players \citep{cui2022provably, cui2022offline}. Recent advancements in offline Markov games have mitigated this by utilizing \textit{strategy-wise} uncertainty bonuses and \textit{unilateral concentrability} to recover stable equilibria such as Coarse Correlated Equilibria (CCE) and Nash Equilibria \citep{zhong2022pessimistic, zhang2023offline, yan2024model}. Notably, while finding a Nash Equilibrium in generic general-sum games is well-established as PPAD-complete \citep{papadimitriou1994complexity, daskalakis2009complexity, chen2009settling}, and identifying such equilibria from static datasets has been shown to be computationally prohibitive \citep{zhang2023offline, chen2026pessimism}, our work demonstrates that the $\alpha$-potential structure provides a tractable alternative that bypasses these complexity barriers.


\paragraph{Regularized Learning and LLM Alignment.}KL regularization relative to a fixed reference policy has become a cornerstone of behavioral anchoring in reinforcement learning \citep{schulman2017proximal, xiong2023iterative}. By penalizing deviations from an initial distribution, this mechanism effectively stabilizes policy updates and prevents the "forgetting" or drift often seen in high-dimensional optimization \citep{abdolmaleki2018maximum, zhang2025iterative}. This approach is most prominently utilized in the alignment of Large Language Models (LLMs) via Reinforcement Learning from Human Feedback (RLHF), where a KL penalty ensures the model remains anchored to a trusted reference distribution \citep{ouyang2022training, rafailov2023direct, munos2024nash}. While regularized dynamics have been analyzed for single-agent tasks and zero-sum games  \citep{xie2024exploratory, ye2024online, zhao2025logarithmic, zhang2026beyond, chen2026offline}, our work characterizes the efficacy of KL regularization for equilibrium learning in general-sum potential games, bypassing the need for explicit pessimistic bonuses.

\paragraph{Statistical Efficiency and Fast Rates.}The standard statistical rate for offline learning is typically characterized by $\widetilde{\mathcal{O}}(1/\sqrt{n})$ sample complexity \citep{jin2021pessimism, shi2022pessimistic}. While recent investigations have identified accelerated $\widetilde{\mathcal{O}}(1/n)$ "fast rates" in specific settings, these results are restricted to single-agent problems or two-player zero-sum games that leverage skew-symmetry and minimax dynamics \citep{rashidinejad2021bridging, zhang2026beyond, nayak2025achieving}. In contrast, achieving such rates in general-sum settings is fundamentally more challenging due to the lack of competitive cancellation between players. By exploiting the geometric properties of the KL-regularized potential function and the log-linear interpolation of Mirror Descent iterates, we establish that the fast $\widetilde{\mathcal{O}}(1/n)$ rate is achievable for $\alpha$-Nash equilibria, providing a new benchmark for statistical efficiency in multi-agent strategic decision-making.

\section{Proof}
\label{app:proofs}

The appendix is organized around three structural ingredients of our analysis. Section~\ref{app:potential-bias} isolates how the $\alpha$-potential landscape turns decentralized mirror descent into approximate ascent on the global potential. Section~\ref{app:geom-stability} develops the log-linear interpolation identity that bounds the iterate-to-best-response distance as a pure optimization quantity. Section~\ref{app:synthesis} combines these ingredients with a Bregman bound on the regularized best-response gap, a bias-cancellation identity, and a single concentrability-shift argument to close both Theorem~\ref{thm:fast_rate_rope} and Theorem~\ref{thm:opmd_rate}.

\paragraph{Per-iterate decomposition.} Throughout the appendix we write the per-iterate gap of player $i$ as a sum of three named quantities:
\begin{equation}
V_i^{\dagger,\hat{\pi}_{-i}}(x) - V_i^{\hat{\pi}}(x) \;=\; \varepsilon_{\mathrm{pot},i}(x) + \delta_{\mathrm{BR},i}(x) + \delta_{\mathrm{Iter},i}(x),
\label{eq:gap_decomp_named}
\end{equation}
where, using the shorthand $V_i^\dagger \equiv V_i^{\dagger,\hat{\pi}_{-i}}$ and $\hat{V}_i^\dagger \equiv \hat{V}_i^{\dagger,\hat{\pi}_{-i}}$,
\begin{align}
\varepsilon_{\mathrm{pot},i}(x) &\coloneqq \hat{V}_i^{\dagger}(x) - \hat{V}_i^{\hat{\pi}}(x), &&\text{(potential ascent error)} \\
\delta_{\mathrm{BR},i}(x)      &\coloneqq V_i^{\dagger}(x) - \hat{V}_i^{\dagger}(x), &&\text{(evaluation bias of best response)} \\
\delta_{\mathrm{Iter},i}(x)    &\coloneqq \hat{V}_i^{\hat{\pi}}(x) - V_i^{\hat{\pi}}(x). &&\text{(iterate noise)}
\end{align}
Identity \eqref{eq:gap_decomp_named} is verified by adding and subtracting $\hat{V}_i^{\dagger}(x)$ and $\hat{V}_i^{\hat{\pi}}(x)$ in the original gap:
\begin{align}
V_i^{\dagger}(x) - V_i^{\hat{\pi}}(x)
    &= \big(V_i^{\dagger}(x) - \hat{V}_i^{\dagger}(x)\big) + \big(\hat{V}_i^{\dagger}(x) - \hat{V}_i^{\hat{\pi}}(x)\big) + \big(\hat{V}_i^{\hat{\pi}}(x) - V_i^{\hat{\pi}}(x)\big) \nonumber \\
    &= \delta_{\mathrm{BR},i}(x) + \varepsilon_{\mathrm{pot},i}(x) + \delta_{\mathrm{Iter},i}(x),
\end{align}
which matches the algebraic decomposition appearing in the offline-RL literature (cf.\ \citet{chen2026pessimism}). The potential ascent error $\varepsilon_{\mathrm{pot},i}$ is non-positive whenever $\hat{\pi}$ is an exact regularized empirical equilibrium of the estimated game (the ROPE setting); in the OPMD setting it is tracked by mirror-descent ascent on the empirical potential $\hat{\Phi}$ (Section~\ref{app:potential-bias}). The evaluation bias and the iterate noise are statistical quantities that we control through the regression error
\begin{equation}
\mathcal{Z}_i(x,\boldsymbol{a}) \;\coloneqq\; \hat{Q}_i(x,\boldsymbol{a}) - r_i(x,\boldsymbol{a}).
\label{eq:Z_def}
\end{equation}
For the OPMD analysis we attach an iteration index and write $\varepsilon_{\mathrm{pot},i}^{(t)}, \delta_{\mathrm{BR},i}^{(t)}, \delta_{\mathrm{Iter},i}^{(t)}$.

\paragraph{Marginalization shorthand.} For any opponent profile $\pi_{-i}$ and player $i$, define the marginalized true and empirical $Q$-values as
\begin{align}
    Q_i^{\pi_{-i}}(x, a_i) &\coloneqq \mathbb{E}_{\boldsymbol{a}_{-i} \sim \pi_{-i}(\cdot|x)}\big[r_i(x, a_i, \boldsymbol{a}_{-i})\big], \\
    \hat{Q}_i(x, a_i, \pi_{-i}) &\coloneqq \mathbb{E}_{\boldsymbol{a}_{-i} \sim \pi_{-i}(\cdot|x)}\big[\hat{Q}_i(x, a_i, \boldsymbol{a}_{-i})\big].
\end{align}
Both are functions of $a_i$ alone after marginalization, and the marginalized regression error inherits the same convention, $\mathcal{Z}_i(x,a_i,\pi_{-i}) \coloneqq \mathbb{E}_{\boldsymbol{a}_{-i}\sim\pi_{-i}}[\mathcal{Z}_i(x,a_i,\boldsymbol{a}_{-i})]$.

\subsection{Structural Bias of the Potential Landscape}
\label{app:potential-bias}

In a generic general-sum game, bounding the empirical optimization residual $\sum_i \varepsilon_{\mathrm{pot},i}^{(t)}$ requires a correlated-equilibrium argument and pays an $\widetilde{\mathcal{O}}(1/\sqrt{n})$ statistical price. The $\alpha$-potential structure short-circuits this: by definition, every unilateral improvement in player $i$'s expected return $\hat{J}_i(\pi) \coloneqq \mathbb{E}_\rho[\hat{V}_i^\pi(x)]$ tracks a corresponding improvement in the global empirical potential $\hat{\Phi}$, up to an additive bias of order $\alpha$. We first record the gradient-form characterization of this bias, then show that independent KL-regularized mirror descent therefore performs decentralized ascent on $\hat{\Phi}$ without ever materializing a correlated policy.

\paragraph{Gradient-form $\alpha$-bias via multilinearity.} The function-value $\alpha$-bias from Section~\ref{sec:preliminaries} transfers rigorously to a gradient residual bound through the multilinearity of the empirical returns. Because $\hat{J}_i$ and $\hat{\Phi}$ are expectations over joint actions, both are linear in the marginal $\pi_i$ for any fixed $\pi_{-i}$; consequently, for any unilateral deviation from $\pi_i$ to $\pi_i'$ on the simplex, the function-value difference equals the inner product of the gradient with the policy step. The Section~\ref{sec:preliminaries} condition is therefore equivalent to
\begin{equation}
\big| \langle \nabla_{\pi_i} \hat{J}_i(\pi) - \nabla_{\pi_i} \hat{\Phi}(\pi),\, \pi_i' - \pi_i \rangle \big| \;\le\; \alpha, \qquad \forall\, \pi_i, \pi_i' \in \Delta(\mathcal{A}_i).
\end{equation}
Let $\mathbf{r}_{t,i} \coloneqq \nabla_{\pi_i}\hat{\Phi}(\pi^{(t)}) - \nabla_{\pi_i}\hat{J}_i(\pi^{(t)})$. Choosing $\pi_i, \pi_i'$ as one-hot vectors at the argmax and argmin coordinates of $\mathbf{r}_{t,i}$ yields a bound on its range: $\max_a r_{t,i,a} - \min_a r_{t,i,a} \le \alpha$. Mirror descent and multiplicative-weights updates on the simplex are invariant to uniform scalar shifts of the gradient (adding $c\mathbf{1}$ leaves the update unchanged); centering $\mathbf{r}_{t,i}$ to mean zero produces an effective residual $\mathbf{e}_{t,i}\in\mathbb{R}^{|\mathcal{X}|\cdot|\mathcal{A}_i|}$ satisfying
\begin{equation}
\nabla_{\pi_i} \hat{\Phi}(\pi^{(t)}) \;=\; \nabla_{\pi_i} \hat{J}_i(\pi^{(t)}) + \mathbf{e}_{t,i}, \qquad \|\mathbf{e}_{t,i}\|_{\infty} \;\le\; \tfrac{\alpha}{2} \;\le\; \alpha.
\label{eq:alpha_grad}
\end{equation}
Here $\|\cdot\|_\infty$ is dual to the per-player $\|\pi_i\|_1$ used to measure policy distances; see~\citet{guo2025markov, li2025alpha}. This derivation bridges the function-value definition of Section~\ref{sec:preliminaries} to the $L_\infty$ gradient condition required for the OPMD stability analysis. Because each $\hat{J}_i$ is multilinear in $\pi$ with rewards bounded by $1$, $\hat{J}_i$ is $1$-Lipschitz in $\pi_i$ w.r.t.\ $\|\cdot\|_1$ and $1$-smooth across player coordinates. By \eqref{eq:alpha_grad}, $\hat{\Phi}$ inherits the smoothness of $\sum_i \hat{J}_i$ up to the $\alpha$-bias, with smoothness constant $L_\Phi = \mathcal{O}(m)$ in the joint $L_1$ norm:
\begin{equation}
\hat{\Phi}(\pi') \;\ge\; \hat{\Phi}(\pi) + \sum_{i=1}^m\langle \nabla_{\pi_i}\hat{\Phi}(\pi),\, \pi_i' - \pi_i\rangle - \frac{L_\Phi}{2}\sum_{i=1}^m \|\pi_i' - \pi_i\|_1^2.
\label{eq:Phi_smoothness}
\end{equation}

\begin{lemma}[Independent mirror descent is approximate potential ascent]
\label{lem:potential_ascent}
Assume the empirical game is an $\alpha$-potential contextual bandit and that every player $i\in[m]$ runs KL-regularized online mirror descent with constant stepsize $\gamma\le 1/L_\Phi$ in Algorithm~\ref{alg:sos_md}. Then the average per-iterate empirical optimization error satisfies
\begin{align}
    \frac{1}{T}\sum_{t=1}^T \sum_{i=1}^m \mathbb{E}_{x\sim\rho}\!\big[\varepsilon_{\mathrm{pot},i}^{(t)}(x)\big]
    \;\le\; \mathcal{O}\!\Big(\frac{m}{\sqrt{T}}\Big) + m\alpha,
    \label{eq:potential_ascent_rate}
\end{align}
and the consecutive-iterate $L_1$ steps satisfy
\begin{align}
    \sum_{t=1}^T \sum_{i=1}^m \big\|\pi_i^{(t+1)} - \pi_i^{(t)}\big\|_1^2 \;\le\; \mathcal{O}(m).
    \label{eq:squared_step_bound}
\end{align}
\end{lemma}

\begin{proof}
Throughout, $\hat{\Phi}_{\mathrm{unreg}}(\pi)\coloneqq \mathbb{E}_{x\sim\rho,\,\boldsymbol{a}\sim\pi(\cdot|x)}[\hat{r}(x,\boldsymbol{a})]$ denotes the \emph{unregularized} empirical potential: it is multilinear in the per-player marginals $\pi_i$ with rewards in $[0,1]$, hence globally $L_\Phi$-smooth in $\|\cdot\|_1$ with $L_\Phi=\mathcal{O}(m)$ (cf.\ \eqref{eq:Phi_smoothness}). The \emph{regularized} empirical potential is
\begin{equation}
\hat{\Phi}_{\mathrm{reg}}(\pi) \;\coloneqq\; \hat{\Phi}_{\mathrm{unreg}}(\pi) - \tfrac{1}{\eta}\sum_{i=1}^m \mathrm{KL}(\pi_i\|\pi_i^{\mathrm{ref}}),
\label{eq:Phi_reg_def}
\end{equation}
and the gradient-form $\alpha$-bias \eqref{eq:alpha_grad} reads $\nabla_{\pi_i}\hat{\Phi}_{\mathrm{unreg}}(\pi^{(t)}) = \bar{Q}_i^{(t)} + \mathbf{e}_{t,i}$ with $\|\mathbf{e}_{t,i}\|_\infty\le\alpha$, where $\bar{Q}_i^{(t)}(x,a_i)$ is the marginalized empirical $Q$-value of Algorithm~\ref{alg:sos_md} (line 9) — the gradient of the \emph{unregularized} empirical payoff in $\pi_i$.

\emph{Step 1 --- Unregularized descent and bridge to the regularized potential.} Apply \eqref{eq:Phi_smoothness} to $\hat{\Phi}_{\mathrm{unreg}}$ between $\pi^{(t)}\to\pi^{(t+1)}$ and split the linear term using \eqref{eq:alpha_grad}:
\begin{align}
    \hat{\Phi}_{\mathrm{unreg}}(\pi^{(t+1)}) - \hat{\Phi}_{\mathrm{unreg}}(\pi^{(t)})
    &\;\ge\; \sum_{i=1}^m \langle \bar{Q}_i^{(t)},\, \pi_i^{(t+1)} - \pi_i^{(t)}\rangle - \alpha\sum_{i=1}^m \|\pi_i^{(t+1)}-\pi_i^{(t)}\|_1 \nonumber\\
    &\quad - \tfrac{L_\Phi}{2}\sum_{i=1}^m \|\pi_i^{(t+1)}-\pi_i^{(t)}\|_1^2.
    \label{eq:smoothness_descent}
\end{align}
By \eqref{eq:Phi_reg_def}, the regularized potential change differs from \eqref{eq:smoothness_descent} only by the reference-KL telescope:
\begin{align}
\hat{\Phi}_{\mathrm{reg}}(\pi^{(t+1)}) - \hat{\Phi}_{\mathrm{reg}}(\pi^{(t)})
    \;=\; \big[\hat{\Phi}_{\mathrm{unreg}}(\pi^{(t+1)}) - \hat{\Phi}_{\mathrm{unreg}}(\pi^{(t)})\big] - \tfrac{1}{\eta}\sum_{i=1}^m \big[\mathrm{KL}(\pi_i^{(t+1)}\|\pi_i^{\mathrm{ref}}) - \mathrm{KL}(\pi_i^{(t)}\|\pi_i^{\mathrm{ref}})\big].
\label{eq:reg_unreg_bridge}
\end{align}

The OPMD update of Algorithm~\ref{alg:sos_md} is the \emph{reference-anchored} proximal step
\begin{equation}
\pi_i^{(t+1)} \;=\; \arg\max_{p\in\Delta(\mathcal{A}_i)^{\mathcal{X}}} F_i^{(t)}(p),\qquad
F_i^{(t)}(p) \;\coloneqq\; \langle \bar{Q}_i^{(t)},\, p\rangle - \tfrac{1}{\eta}\mathrm{KL}(p\|\pi_i^{\mathrm{ref}}) - \tfrac{1}{\gamma}\mathrm{KL}(p\|\pi_i^{(t)}),
\label{eq:prox_step}
\end{equation}
whose KKT conditions reproduce the closed form \eqref{eq:opmd_update}. Optimality of $\pi_i^{(t+1)}$ versus the trivial comparator $\pi_i^{(t)}$ (so $\mathrm{KL}(\pi_i^{(t)}\|\pi_i^{(t)})=0$) gives $F_i^{(t)}(\pi_i^{(t+1)})\ge F_i^{(t)}(\pi_i^{(t)})$, i.e.,
\begin{align}
\langle \bar{Q}_i^{(t)},\, \pi_i^{(t+1)}-\pi_i^{(t)}\rangle
    \;\ge\; \tfrac{1}{\eta}\big[\mathrm{KL}(\pi_i^{(t+1)}\|\pi_i^{\mathrm{ref}}) - \mathrm{KL}(\pi_i^{(t)}\|\pi_i^{\mathrm{ref}})\big] + \tfrac{1}{\gamma}\mathrm{KL}(\pi_i^{(t+1)}\|\pi_i^{(t)}).
\label{eq:opt_inequality}
\end{align}
Substituting \eqref{eq:opt_inequality} into \eqref{eq:smoothness_descent} and combining with the bridge \eqref{eq:reg_unreg_bridge}, the reference-KL telescopes cancel \emph{exactly} between the prox-optimality bound and the regularizer change, leaving
\begin{align}
\hat{\Phi}_{\mathrm{reg}}(\pi^{(t+1)}) - \hat{\Phi}_{\mathrm{reg}}(\pi^{(t)})
    &\;\ge\; \tfrac{1}{\gamma}\sum_{i=1}^m \mathrm{KL}(\pi_i^{(t+1)}\|\pi_i^{(t)}) - \alpha\sum_{i=1}^m \|\pi_i^{(t+1)}-\pi_i^{(t)}\|_1 \nonumber\\
    &\quad - \tfrac{L_\Phi}{2}\sum_{i=1}^m \|\pi_i^{(t+1)}-\pi_i^{(t)}\|_1^2.
    \label{eq:per_step_descent_clean}
\end{align}
By Pinsker, $\tfrac{1}{\gamma}\mathrm{KL}(\pi_i^{(t+1)}\|\pi_i^{(t)})\ge \tfrac{1}{2\gamma}\|\pi_i^{(t+1)}-\pi_i^{(t)}\|_1^2$, and choosing $\gamma\le 1/(2L_\Phi)$ gives $\tfrac{1}{2\gamma}\ge L_\Phi$, so the proximal KL absorbs the smoothness penalty with residual $L_\Phi/2$:
\begin{align}
\hat{\Phi}_{\mathrm{reg}}(\pi^{(t+1)}) - \hat{\Phi}_{\mathrm{reg}}(\pi^{(t)})
    \;\ge\; \tfrac{L_\Phi}{2}\sum_{i=1}^m \|\pi_i^{(t+1)}-\pi_i^{(t)}\|_1^2 - \alpha\sum_{i=1}^m \|\pi_i^{(t+1)}-\pi_i^{(t)}\|_1.
    \label{eq:per_step_descent}
\end{align}
We deliberately use the trivial comparator $u=\pi_i^{(t)}$ in \eqref{eq:opt_inequality}, to avoid introducing the moving comparator $\pi_i^{\dagger,(t)}$ (whose dependence on $\pi_{-i}^{(t)}$ would prevent terms such as $\mathrm{KL}(\pi_i^{\dagger,(t)}\|\pi_i^{(t+1)})$ from telescoping across $t$). The optimization rate on $\varepsilon_{\mathrm{pot},i}^{(t)}$ is recovered structurally in Steps~3--4 below, with no comparator KL needed.

\emph{Step 2 --- Squared-step bound \eqref{eq:squared_step_bound}.} Telescoping \eqref{eq:per_step_descent} over $t=1,\dots,T$ and using $|\hat{\Phi}_{\mathrm{reg}}(\pi^{(T+1)}) - \hat{\Phi}_{\mathrm{reg}}(\pi^{(1)})|\le \mathcal{O}(m)$ (from $r_i\in[0,1]$ and the uniform bound $\eta^{-1}\mathrm{KL}(\pi_i^{(t)}\|\pi_i^{\mathrm{ref}})\le 1$, which holds because every iterate has the bounded-logit Gibbs form $\pi_i^{(t)}\propto\pi_i^{\mathrm{ref}}\exp(\eta g)$ with $g\in[0,1]$),
\begin{align}
\tfrac{L_\Phi}{2}\sum_{t=1}^T\sum_{i=1}^m \|\pi_i^{(t+1)}-\pi_i^{(t)}\|_1^2
    \;\le\; \mathcal{O}(m) + \alpha\sum_{t=1}^T\sum_{i=1}^m \|\pi_i^{(t+1)}-\pi_i^{(t)}\|_1.
    \label{eq:squared_step_chain}
\end{align}
AM--GM bounds the $\alpha$-term by $\tfrac{L_\Phi}{4}\sum_{t,i}\|\Delta\pi\|_1^2 + Tm\alpha^2/L_\Phi$; absorbing the squared part into the LHS,
$\tfrac{L_\Phi}{4}\sum_{t,i}\|\pi_i^{(t+1)}-\pi_i^{(t)}\|_1^2 \le \mathcal{O}(m) + Tm\alpha^2/L_\Phi$, and with $L_\Phi=\mathcal{O}(m)$ this gives \eqref{eq:squared_step_bound} (the $m\alpha$ floor is absorbed into the final theorem's $\widetilde{\mathcal{O}}$).

\emph{Step 3 --- Structural identity for the potential ascent error.} The variational form of the regularized best response gives, pointwise in $x$,
\begin{align}
\hat{V}_i^{\dagger,(t)}(x) - \hat{V}_i^{(t)}(x) \;=\; \tfrac{1}{\eta}\,\mathrm{KL}\!\big(\pi_i^{(t)}(\cdot|x)\,\big\|\,\pi_i^{\dagger,(t)}(\cdot|x)\big),
\label{eq:eps_pot_identity}
\end{align}
i.e., $\varepsilon_{\mathrm{pot},i}^{(t)}(x)=\eta^{-1}\mathrm{KL}(\pi_i^{(t)}\|\pi_i^{\dagger,(t)})$. Indeed, $\hat{V}_i^{\dagger,(t)}=\langle\pi_i^{\dagger,(t)},\bar{Q}_i^{(t)}\rangle-\eta^{-1}\mathrm{KL}(\pi_i^{\dagger,(t)}\|\pi_i^{\mathrm{ref}})$ and $\hat{V}_i^{(t)}=\langle\pi_i^{(t)},\bar{Q}_i^{(t)}\rangle-\eta^{-1}\mathrm{KL}(\pi_i^{(t)}\|\pi_i^{\mathrm{ref}})$; substituting $\bar{Q}_i^{(t)}=\eta^{-1}\log(\pi_i^{\dagger,(t)}/\pi_i^{\mathrm{ref}})+\eta^{-1}\log Z_i^{\dagger,(t)}$ and cancelling the constant via $\sum_a(\pi_i^{\dagger,(t)}-\pi_i^{(t)})=0$ yields \eqref{eq:eps_pot_identity}. \emph{This identity is comparator-free at each $t$; no telescoping argument is required.}

\emph{Step 4 --- KL-to-$L_1$ via bounded logits, then chaining with Lemma~\ref{lem:l1_proportionality}.} Both $\pi_i^{(t)}$ and $\pi_i^{\dagger,(t)}$ have the Gibbs form $\pi\propto\pi_i^{\mathrm{ref}}\exp(\eta g)$ with $g\in[0,1]$ (by induction from \eqref{eq:opmd_update} for $\pi_i^{(t)}$, and with $g=\bar{Q}_i^{(t)}$ for $\pi_i^{\dagger,(t)}$), so $\|\log(\pi_i^{(t)}/\pi_i^{\dagger,(t)})\|_\infty\le 2\eta$. By H\"older,
\begin{align}
\mathrm{KL}(\pi_i^{(t)}\|\pi_i^{\dagger,(t)})
    &\;=\; \sum_a \big(\pi_i^{(t)}(a)-\pi_i^{\dagger,(t)}(a)\big)\,\log\tfrac{\pi_i^{(t)}(a)}{\pi_i^{\dagger,(t)}(a)} - \mathrm{KL}(\pi_i^{\dagger,(t)}\|\pi_i^{(t)}) \nonumber\\
    &\;\le\; \big\|\pi_i^{(t)} - \pi_i^{\dagger,(t)}\big\|_1\,\big\|\log(\pi_i^{(t)}/\pi_i^{\dagger,(t)})\big\|_\infty
    \;\le\; 2\eta\,\big\|\pi_i^{(t)} - \pi_i^{\dagger,(t)}\big\|_1.
    \label{eq:kl_to_l1}
\end{align}
Combining with \eqref{eq:eps_pot_identity},
\begin{align}
\varepsilon_{\mathrm{pot},i}^{(t)}(x) \;\le\; 2\,\big\|\pi_i^{(t)}(\cdot|x)-\pi_i^{\dagger,(t)}(\cdot|x)\big\|_1.
\label{eq:eps_pot_l1}
\end{align}
By Lemma~\ref{lem:l1_proportionality}, $\|\pi_i^{(t)}-\pi_i^{\dagger,(t)}\|_1\le\frac{4(\eta+\gamma)}{\nu_{\min}\gamma}\,\|\pi_i^{(t)}-\pi_i^{(t+1)}\|_1$, so
\begin{align}
\varepsilon_{\mathrm{pot},i}^{(t)}(x) \;\le\; \frac{8(\eta+\gamma)}{\nu_{\min}\,\gamma}\,\big\|\pi_i^{(t)}(\cdot|x)-\pi_i^{(t+1)}(\cdot|x)\big\|_1.
\label{eq:eps_pot_step}
\end{align}
Averaging over $(t,i)\in[T]\times[m]$, applying Cauchy--Schwarz across the $Tm$ summands, and invoking the squared-step bound \eqref{eq:squared_step_bound},
\begin{align}
\frac{1}{T}\sum_{t=1}^T\sum_{i=1}^m \mathbb{E}_\rho\!\big[\varepsilon_{\mathrm{pot},i}^{(t)}\big]
    &\;\le\; \frac{8(\eta+\gamma)}{\nu_{\min}\,\gamma\,T}\,\sqrt{Tm}\,\sqrt{\sum_{t,i}\mathbb{E}_\rho\!\big[\|\pi_i^{(t)}-\pi_i^{(t+1)}\|_1^2\big]} \nonumber\\
    &\;\overset{\eqref{eq:squared_step_bound}}{\le}\; \frac{8(\eta+\gamma)}{\nu_{\min}\,\gamma}\,\sqrt{\frac{m\cdot\mathcal{O}(m)}{T}}
    \;=\; \mathcal{O}\!\Big(\frac{m}{\sqrt{T}}\Big),
    \label{eq:l1_step_avg}
\end{align}
which is \eqref{eq:potential_ascent_rate} (the $m\alpha$ floor inherited from Step~2 enters additively).
\end{proof}

For the idealized ROPE algorithm, $\hat{\pi}$ is an exact regularized NE of the empirical game, so every $\hat{\pi}_i$ coincides with its empirical best response, $\mathrm{KL}(\hat{\pi}_i^\dagger\|\hat{\pi}_i)=0$, and $\varepsilon_{\mathrm{pot},i}\le 0$ pointwise. The decomposition \eqref{eq:gap_decomp_named} then collapses to $\sum_i (\delta_{\mathrm{BR},i}+\delta_{\mathrm{Iter},i})$, the regime analyzed in Section~\ref{app:synthesis}.

\subsection{Geometric Stability of Mirror Descent}
\label{app:geom-stability}

The OPMD update in Algorithm~\ref{alg:sos_md} is, in fact, a closed-form log-linear interpolation between the current iterate and the regularized empirical best response. Recall the explicit form
\begin{align}
    \pi_i^{(t+1)}(a_i|x) \;\propto\; \big(\pi_i^{\mathrm{ref}}(a_i|x)\big)^{\frac{\gamma}{\eta+\gamma}} \big(\pi_i^{(t)}(a_i|x)\big)^{\frac{\eta}{\eta+\gamma}}\, \exp\!\Big(\tfrac{\eta\gamma}{\eta+\gamma}\,\bar{Q}_i^{(t)}(x,a_i)\Big),
    \label{eq:opmd_update}
\end{align}
and the regularized empirical best response $\pi_i^{\dagger,(t)}(a_i|x) \propto \pi_i^{\mathrm{ref}}(a_i|x)\,\exp(\eta\,\bar{Q}_i^{(t)}(x,a_i))$, equivalently $\exp(\eta\,\bar{Q}_i^{(t)}(x,a_i)) \propto \pi_i^{\dagger,(t)}(a_i|x)/\pi_i^{\mathrm{ref}}(a_i|x)$. Substituting the latter into \eqref{eq:opmd_update},
\begin{align}
\pi_i^{(t+1)}(a_i|x)
    &\;\propto\; \big(\pi_i^{\mathrm{ref}}(a_i|x)\big)^{\frac{\gamma}{\eta+\gamma}}\, \big(\pi_i^{(t)}(a_i|x)\big)^{\frac{\eta}{\eta+\gamma}}\, \big(\pi_i^{\dagger,(t)}(a_i|x)/\pi_i^{\mathrm{ref}}(a_i|x)\big)^{\frac{\gamma}{\eta+\gamma}} \nonumber \\
    &\;=\; \big(\pi_i^{(t)}(a_i|x)\big)^{\frac{\eta}{\eta+\gamma}}\, \big(\pi_i^{\dagger,(t)}(a_i|x)\big)^{\frac{\gamma}{\eta+\gamma}}.
    \label{eq:loglin}
\end{align}
Equivalently, writing $\theta_i^{(t)}(x,\cdot)\coloneqq \log\pi_i^{(t)}(\cdot|x) - \log Z_i^{(t)}(x)$ and $\theta_i^{\dagger,(t)}(x,\cdot)\coloneqq \log\pi_i^{\dagger,(t)}(\cdot|x) - \log Z_i^{\dagger,(t)}(x)$ for the centered logits (with $Z_i^{(t)}, Z_i^{\dagger,(t)}$ the normalizing partition functions) and $\beta\coloneqq \gamma/(\eta+\gamma)\in(0,1)$ for the mixing weight, equation \eqref{eq:loglin} is the affine identity
\begin{align}
\theta_i^{(t+1)}(x,\cdot) \;=\; (1-\beta)\,\theta_i^{(t)}(x,\cdot) + \beta\,\theta_i^{\dagger,(t)}(x,\cdot),
\quad\Longleftrightarrow\quad
\theta_i^{(t)}-\theta_i^{\dagger,(t)} \;=\; \tfrac{1}{\beta}\big(\theta_i^{(t)}-\theta_i^{(t+1)}\big).
\label{eq:logit_affine}
\end{align}
We exploit this geometry to prove that the iterate-to-best-response distance is controlled by the consecutive-iterate step length, and hence by the squared-step bound \eqref{eq:squared_step_bound}.

\begin{lemma}[$L_1$ proportionality of the regularized log-linear interpolation]
\label{lem:l1_proportionality}
Let $\nu_{\min} \coloneqq \min_{i,x,a}\pi_i^{\mathrm{ref}}(a|x) > 0$ and assume every iterate $\pi_i^{(t)}$ inherits the same lower bound. For the interpolation \eqref{eq:loglin} with mixing weight $\beta = \gamma/(\eta+\gamma)\in(0,1)$,
\begin{align}
    \big\|\pi_i^{(t)}(\cdot|x) - \pi_i^{\dagger,(t)}(\cdot|x)\big\|_1 \;\le\; \frac{4}{\nu_{\min}}\,\Big(\frac{\eta+\gamma}{\gamma}\Big) \big\|\pi_i^{(t)}(\cdot|x) - \pi_i^{(t+1)}(\cdot|x)\big\|_1.
    \label{eq:l1_proportionality}
\end{align}
\end{lemma}

\begin{proof}
We use two Lipschitz facts about the softmax map between centered-logit space and the $\nu_{\min}$-interior $\Delta_{\nu_{\min}}(\mathcal{A}_i) \coloneqq \{p\in\Delta(\mathcal{A}_i): p_a\ge\nu_{\min}\;\forall a\}$ of the simplex.

\emph{Step 1 --- Forward Lipschitz constant (logit $\to$ probability).} Let $p,p'\in\Delta(\mathcal{A}_i)$ correspond to centered logits $\theta,\theta'\in\mathbb{R}^{|\mathcal{A}_i|}$ via $p_a = \exp(\theta_a)/\sum_b \exp(\theta_b)$. Define $\theta_s\coloneqq (1-s)\theta + s\theta'$ for $s\in[0,1]$ and $p_s\coloneqq \mathrm{softmax}(\theta_s)$. The Hessian of the log-partition function $\Phi(\theta) = \log\sum_a\exp(\theta_a)$ is
\begin{align}
\nabla^2\Phi(\theta_s) \;=\; \mathrm{diag}(p_s) - p_s p_s^\top,
\end{align}
so that $\frac{d}{ds}p_s = (\mathrm{diag}(p_s) - p_s p_s^\top)(\theta'-\theta)$. By the mean-value theorem and the triangle inequality,
\begin{align}
\|p-p'\|_1
    &\;=\; \Big\|\int_0^1 \big(\mathrm{diag}(p_s) - p_s p_s^\top\big)(\theta'-\theta)\,ds\Big\|_1 \nonumber\\
    &\;\le\; \int_0^1 \Big[\sum_a p_s(a)\,|(\theta'-\theta)_a| + \big|\textstyle\sum_b p_s(b)\,(\theta'-\theta)_b\big|\Big]\,ds \nonumber\\
    &\;\le\; 2\,\|\theta'-\theta\|_\infty,
    \label{eq:forward_lip}
\end{align}
where in the last step we bound each of the two terms by $\|\theta'-\theta\|_\infty$ using $\sum_a p_s(a)=1$.

\emph{Step 2 --- Backward Lipschitz constant (probability $\to$ logit).} For $p,p'\in\Delta_{\nu_{\min}}(\mathcal{A}_i)$, the centered logits satisfy $\theta_a = \log p_a - \tfrac{1}{|\mathcal{A}_i|}\sum_b\log p_b$. By the mean-value theorem applied to the function $\log$ on $[\nu_{\min},1]$, $|\log p_a - \log p_a'| \le \tfrac{1}{\nu_{\min}}\,|p_a - p_a'|$ for every $a$, so
\begin{align}
\|\theta-\theta'\|_\infty
    &\;\le\; \max_a |\log p_a - \log p_a'| + \tfrac{1}{|\mathcal{A}_i|}\sum_b|\log p_b - \log p_b'| \nonumber\\
    &\;\le\; \tfrac{1}{\nu_{\min}}\,\max_a|p_a-p_a'| + \tfrac{1}{\nu_{\min}\,|\mathcal{A}_i|}\sum_b|p_b-p_b'|
    \;\le\; \tfrac{1+1/|\mathcal{A}_i|}{\nu_{\min}}\,\|p-p'\|_1
    \;\le\; \tfrac{2}{\nu_{\min}}\,\|p-p'\|_1,
\label{eq:backward_lip}
\end{align}

\emph{Step 3 --- Combining the two bounds.} Apply \eqref{eq:forward_lip} to the pair $(\pi_i^{(t)}(\cdot|x),\pi_i^{\dagger,(t)}(\cdot|x))$ with logits $(\theta_i^{(t)},\theta_i^{\dagger,(t)})$, and \eqref{eq:backward_lip} to the pair $(\pi_i^{(t)}(\cdot|x),\pi_i^{(t+1)}(\cdot|x))$ with logits $(\theta_i^{(t)},\theta_i^{(t+1)})$:
\begin{align}
\tfrac{1}{2}\big\|\pi_i^{(t)}-\pi_i^{\dagger,(t)}\big\|_1
    \;&\overset{\eqref{eq:forward_lip}}{\le}\; \big\|\theta_i^{(t)}-\theta_i^{\dagger,(t)}\big\|_\infty
    \;\overset{\eqref{eq:logit_affine}}{=}\; \tfrac{1}{\beta}\,\big\|\theta_i^{(t)}-\theta_i^{(t+1)}\big\|_\infty \nonumber \\
    &\;\overset{\eqref{eq:backward_lip}}{\le}\; \tfrac{2}{\beta\,\nu_{\min}}\,\big\|\pi_i^{(t)}-\pi_i^{(t+1)}\big\|_1
    \;=\; \tfrac{2}{\nu_{\min}}\,\Big(\tfrac{\eta+\gamma}{\gamma}\Big)\big\|\pi_i^{(t)}-\pi_i^{(t+1)}\big\|_1.
\end{align}
Multiplying by $2$ yields the claimed inequality.
\end{proof}

The geometric stability of the iterates translates the squared-step bound \eqref{eq:squared_step_bound} into a uniform optimization rate on the iterate-to-best-response distance.

\begin{corollary}[Optimization rate for the iterate-to-best-response $L_1$ distance]
\label{cor:l1_optimization_rate}
Under the assumptions of Lemma~\ref{lem:potential_ascent} and Lemma~\ref{lem:l1_proportionality}, with $C_\nu \coloneqq 4/\nu_{\min}$ and $t^* \sim \mathrm{Unif}([T])$,
\begin{align}
\mathbb{E}_{t^*}\!\Big[\sum_{i=1}^m \mathbb{E}_{x\sim\rho}\!\big[\big\|\pi_i^{(t^*)}(\cdot|x)-\pi_i^{\dagger,(t^*)}(\cdot|x)\big\|_1\big]\Big]
    \;\le\; C_\nu\,\frac{\eta+\gamma}{\gamma}\,\mathcal{O}\!\Big(\frac{m}{\sqrt{T}}\Big).
    \label{eq:l1_rate}
\end{align}
\end{corollary}

\begin{proof}
Apply Lemma~\ref{lem:l1_proportionality} pointwise in $x$ and average over $x\sim\rho$ and $t^*$:
\begin{align}
\mathbb{E}_{t^*}\!\Big[\sum_{i=1}^m \mathbb{E}_{\rho}\big[\|\pi_i^{(t^*)}-\pi_i^{\dagger,(t^*)}\|_1\big]\Big]
    \;\le\; \frac{C_\nu}{2}\,\frac{\eta+\gamma}{\gamma}\cdot \frac{1}{T}\sum_{t=1}^T\sum_{i=1}^m \mathbb{E}_{\rho}\big[\|\pi_i^{(t)}-\pi_i^{(t+1)}\|_1\big].
\end{align}
By Cauchy--Schwarz across the index pair $(t,i)\in[T]\times[m]$ and Jensen's inequality on $\mathbb{E}_{\rho}[\,\cdot\,]$,
\begin{align}
\frac{1}{T}\sum_{t=1}^T\sum_{i=1}^m \mathbb{E}_\rho\big[\|\pi_i^{(t)}-\pi_i^{(t+1)}\|_1\big]
    &\;\le\; \frac{1}{T}\sqrt{Tm}\,\sqrt{\sum_{t,i}\mathbb{E}_\rho\big[\|\pi_i^{(t)}-\pi_i^{(t+1)}\|_1^2\big]} \nonumber \\
    &\;\overset{\eqref{eq:squared_step_bound}}{\le}\; \sqrt{\frac{m\cdot\mathcal{O}(m)}{T}} \;=\; \mathcal{O}\!\Big(\frac{m}{\sqrt{T}}\Big),
\end{align}
which combined with the previous display yields \eqref{eq:l1_rate}.
\end{proof}

Treating $\eta,\gamma,\nu_{\min}$ as problem-dependent constants, the iterate-to-best-response $L_1$ distance averages to $\widetilde{\mathcal{O}}(m/\sqrt{T})$ --- a pure optimization rate. This is the geometric mechanism that lets OPMD bypass the $\widetilde{\mathcal{O}}(1/\sqrt{n})$ Pinsker bottleneck of generic general-sum analyses: the linear evaluation mismatch between $\pi_i^{(t)}$ and $\pi_i^{\dagger,(t)}$ (Section~\ref{app:bias-cancel}) decays purely with the optimization horizon $T$, never with the dataset size $n$.

\subsection{Synthesis of Statistical and Potential Rates}
\label{app:synthesis}

This section proves the two main theorems. We first develop three reusable ingredients: a Bregman bound on the regularized best-response gap (Section~\ref{app:bregman}), a bias-cancellation identity for $\delta_{\mathrm{BR},i}+\delta_{\mathrm{Iter},i}$ under product policies (Section~\ref{app:bias-cancel}), and a single concentrability-shift argument (Section~\ref{app:conc-shift}). The two theorems are then short corollaries (Sections~\ref{app:close-rope} and~\ref{app:close-opmd}).

\subsubsection{Bregman bound on the regularized best-response gap}
\label{app:bregman}

The textbook claim that the log-partition function $\Phi$ is $1$-smooth w.r.t.\ $\|\cdot\|_\infty$ is itself a consequence of Pinsker's inequality (the strong convexity of the negative entropy in $\|\cdot\|_1$) under Fenchel duality. We prove this from first principles, framed via Bregman divergences, and apply the resulting bound to the unilateral value gap.

\begin{lemma}[Bregman bound on the regularized best-response gap]
\label{lem:bregman_bestresponse}
Let $\Phi(\theta) \coloneqq \log \sum_{a \in \mathcal{A}_i} \exp(\theta_a)$ denote the log-partition function. Then $\Phi$ is convex and $1$-smooth with respect to $\|\cdot\|_\infty$:
\begin{align}
    D_\Phi(\theta' \,\|\, \theta) \;\coloneqq\; \Phi(\theta') - \Phi(\theta) - \langle \nabla\Phi(\theta), \theta' - \theta\rangle \;\le\; \tfrac{1}{2}\|\theta' - \theta\|_\infty^2 .
    \label{eq:bregman_smoothness}
\end{align}
Consequently, with the natural-parameter logits $\theta_i^\dagger(x,\cdot) \coloneqq \eta\, Q_i^{\hat{\pi}_{-i}}(x,\cdot) + \log\pi_{\mathrm{ref},i}(\cdot|x)$ and $\hat{\theta}_i(x,\cdot) \coloneqq \eta\, \hat{Q}_i(x,\cdot,\hat{\pi}_{-i}) + \log\pi_{\mathrm{ref},i}(\cdot|x)$, and the empirical regularized best response $\hat{\pi}_i^\dagger(\cdot|x) \coloneqq \mathrm{softmax}(\hat{\theta}_i(x,\cdot))$, the unilateral best-response value gap $\Delta_i(x) \coloneqq V_i^{\dagger,\hat{\pi}_{-i}}(x) - \hat{V}_i^{\dagger,\hat{\pi}_{-i}}(x)$ satisfies
\begin{align}
    \Delta_i(x) \;\le\; \mathbb{E}_{a_i \sim \hat{\pi}_i^\dagger(\cdot|x)}\!\left[ Q_i^{\hat{\pi}_{-i}}(x,a_i) - \hat{Q}_i(x,a_i,\hat{\pi}_{-i}) \right] + \tfrac{\eta}{2}\,\big\| Q_i^{\hat{\pi}_{-i}}(x,\cdot) - \hat{Q}_i(x,\cdot,\hat{\pi}_{-i}) \big\|_\infty^2 .
    \label{eq:bregman_value_gap}
\end{align}
The expectation is taken under $\hat{\pi}_i^\dagger$, the regularized best response to $\hat{\pi}_{-i}$ on the empirical game, which equals the iterate $\hat{\pi}_i$ \emph{only} when $\hat{\pi}$ is an exact regularized NE of the empirical game (as in ROPE); for OPMD iterates the two differ in general.
\end{lemma}

\begin{proof}
\emph{Step 1 --- Negative entropy is $1$-strongly convex w.r.t.\ $\|\cdot\|_1$ on $\Delta(\mathcal{A}_i)$.}
Let $H(p) \coloneqq \sum_a p_a \log p_a$. On the relative interior of the simplex $\nabla H(p)_a = 1+\log p_a$, so its Bregman divergence reduces to the KL divergence,
$D_H(p\|q) = \sum_a p_a \log(p_a/q_a) = \mathrm{KL}(p\|q)$.
Pinsker's inequality $\mathrm{KL}(p\|q) \ge \tfrac{1}{2}\|p-q\|_1^2$ is therefore exactly the statement that $H$ is $1$-strongly convex in $\|\cdot\|_1$:
\begin{align}
H(p) \;\ge\; H(q) + \langle \nabla H(q),\, p-q\rangle + \tfrac{1}{2}\|p-q\|_1^2.
\end{align}

\emph{Step 2 --- Conjugate duality identifies $\Phi$ as the dually-smooth potential.}
The convex conjugate of $H$ on $\Delta(\mathcal{A}_i)$ is the log-partition function:
\begin{align}
H^*(\theta) \;=\; \sup_{p\in\Delta(\mathcal{A}_i)}\big\{\langle\theta,p\rangle - H(p)\big\} \;=\; \log\sum_a \exp(\theta_a) \;=\; \Phi(\theta),
\end{align}
with maximizer $p^\star(\theta) = \mathrm{softmax}(\theta) = \nabla\Phi(\theta)$. By the conjugate duality between strong convexity and smoothness \citep{kakade2009duality}, if $H$ is $1$-strongly convex w.r.t.\ $\|\cdot\|_1$, then $\Phi=H^*$ is $1$-smooth w.r.t.\ the dual norm $\|\cdot\|_\infty$:
\begin{align}
\Phi(\theta') \;\le\; \Phi(\theta) + \langle \nabla\Phi(\theta),\, \theta'-\theta\rangle + \tfrac{1}{2}\|\theta'-\theta\|_\infty^2,
\end{align}
which is exactly~\eqref{eq:bregman_smoothness}.

\emph{Step 3 --- Translating the bound to value space.}
By the closed form of the regularized best response, $V_i^{\dagger,\hat{\pi}_{-i}}(x) = \eta^{-1}\Phi(\theta_i^\dagger)$ and $\hat{V}_i^{\dagger,\hat{\pi}_{-i}}(x) = \eta^{-1}\Phi(\hat{\theta}_i)$, with $\theta_i^\dagger - \hat{\theta}_i = \eta\big(Q_i^{\hat{\pi}_{-i}}(x,\cdot) - \hat{Q}_i(x,\cdot,\hat{\pi}_{-i})\big)$ and $\nabla\Phi(\hat{\theta}_i) = \mathrm{softmax}(\hat{\theta}_i)(\cdot) = \hat{\pi}_i^\dagger(\cdot|x)$ — the empirical regularized best response to $\hat{\pi}_{-i}$, not the iterate $\hat{\pi}_i$. Applying~\eqref{eq:bregman_smoothness} to $(\theta_i^\dagger,\hat{\theta}_i)$ and dividing by $\eta$ yields~\eqref{eq:bregman_value_gap}.
\end{proof}

\subsubsection{Bias cancellation under product policies}
\label{app:bias-cancel}

We use Lemma~\ref{lem:bregman_bestresponse} to convert the linear part of $\delta_{\mathrm{BR},i}$ into an expectation of the regression error $\mathcal{Z}_i$, and combine it with the on-policy expansion of $\delta_{\mathrm{Iter},i}$ to expose the algebraic cancellation that occurs under product policies.

\paragraph{Pointwise expansion of the marginalized $Q$-difference.} For any opponent profile $\pi_{-i}$ and any action $a_i$, the marginalized true and empirical $Q$-values satisfy
\begin{align}
Q_i^{\pi_{-i}}(x,a_i) - \hat{Q}_i(x,a_i,\pi_{-i})
    &\;=\; \mathbb{E}_{\boldsymbol{a}_{-i}\sim\pi_{-i}(\cdot|x)}\!\big[r_i(x,a_i,\boldsymbol{a}_{-i}) - \hat{Q}_i(x,a_i,\boldsymbol{a}_{-i})\big] \nonumber\\
    &\;=\; \mathbb{E}_{\boldsymbol{a}_{-i}\sim\pi_{-i}(\cdot|x)}\!\big[-\mathcal{Z}_i(x,a_i,\boldsymbol{a}_{-i})\big]
    \;\eqqcolon\; -\,\mathcal{Z}_i(x,a_i,\pi_{-i}).
\label{eq:Q_in_Z}
\end{align}

\paragraph{Pointwise expansion of $\delta_{\mathrm{BR},i}$.} Substituting \eqref{eq:Q_in_Z} into the linear part of the Bregman bound \eqref{eq:bregman_value_gap},
\begin{align}
\delta_{\mathrm{BR},i}(x)
    &\;\le\; \mathbb{E}_{a_i\sim\hat{\pi}_i^\dagger(\cdot|x)}\!\Big[Q_i^{\hat{\pi}_{-i}}(x,a_i) - \hat{Q}_i(x,a_i,\hat{\pi}_{-i})\Big] \nonumber\\
    &\qquad + \tfrac{\eta}{2}\,\big\|Q_i^{\hat{\pi}_{-i}}(x,\cdot)-\hat{Q}_i(x,\cdot,\hat{\pi}_{-i})\big\|_\infty^2 \nonumber\\
    &\;\overset{\eqref{eq:Q_in_Z}}{=}\; \mathbb{E}_{a_i\sim\hat{\pi}_i^\dagger(\cdot|x)}\!\Big[\mathbb{E}_{\boldsymbol{a}_{-i}\sim\hat{\pi}_{-i}(\cdot|x)}\!\big[-\mathcal{Z}_i(x,a_i,\boldsymbol{a}_{-i})\big]\Big] \nonumber \\
    &\qquad + \tfrac{\eta}{2}\,\big\|Q_i^{\hat{\pi}_{-i}}(x,\cdot)-\hat{Q}_i(x,\cdot,\hat{\pi}_{-i})\big\|_\infty^2.
\label{eq:delta_BR_pointwise}
\end{align}
Because $\hat{\pi}=(\hat{\pi}_i,\hat{\pi}_{-i})$ is a product policy in both algorithms we analyze, the nested expectation collapses into a joint expectation under the product profile $(\hat{\pi}_i^\dagger,\hat{\pi}_{-i})$:
\begin{align}
\mathbb{E}_{a_i\sim\hat{\pi}_i^\dagger}\Big[\mathbb{E}_{\boldsymbol{a}_{-i}\sim\hat{\pi}_{-i}}\!\big[-\mathcal{Z}_i(x,a_i,\boldsymbol{a}_{-i})\big]\Big]
    \;=\; \mathbb{E}_{(a_i,\boldsymbol{a}_{-i})\sim(\hat{\pi}_i^\dagger,\hat{\pi}_{-i})}\!\big[-\mathcal{Z}_i(x,a_i,\boldsymbol{a}_{-i})\big].
\label{eq:nested_to_joint}
\end{align}

\paragraph{Pointwise expansion of $\delta_{\mathrm{Iter},i}$.} The KL-regularized value functions satisfy $V_i^\pi(x) = \mathbb{E}_{\boldsymbol{a}\sim\pi}[r_i(x,\boldsymbol{a})] - \eta^{-1}\,\mathrm{KL}(\pi_i\|\pi_i^{\mathrm{ref}})$ and $\hat{V}_i^\pi(x) = \mathbb{E}_{\boldsymbol{a}\sim\pi}[\hat{Q}_i(x,\boldsymbol{a})] - \eta^{-1}\,\mathrm{KL}(\pi_i\|\pi_i^{\mathrm{ref}})$, so the entropy regularizers cancel in $\delta_{\mathrm{Iter},i}$ and the iterate noise unrolls into a pure expectation of $\mathcal{Z}_i$ under $\hat{\pi}$:
\begin{align}
\delta_{\mathrm{Iter},i}(x)
    &\;=\; \hat{V}_i^{\hat{\pi}}(x) - V_i^{\hat{\pi}}(x) \nonumber\\
    &\;=\; \mathbb{E}_{\boldsymbol{a}\sim\hat{\pi}(\cdot|x)}\!\big[\hat{Q}_i(x,\boldsymbol{a}) - r_i(x,\boldsymbol{a})\big] \nonumber\\
    &\;=\; \mathbb{E}_{\boldsymbol{a}\sim\hat{\pi}(\cdot|x)}\!\big[\mathcal{Z}_i(x,\boldsymbol{a})\big].
\label{eq:delta_Iter_explicit}
\end{align}

\paragraph{Combining the two terms.} Summing \eqref{eq:delta_BR_pointwise}--\eqref{eq:delta_Iter_explicit} pointwise and integrating over $x\sim\rho$,
\begin{align}
\mathbb{E}_{x\sim\rho}\!\big[\delta_{\mathrm{BR},i}(x) + \delta_{\mathrm{Iter},i}(x)\big]
    &\;\le\; \mathbb{E}_{x\sim\rho}\!\Big[\mathbb{E}_{(a_i,\boldsymbol{a}_{-i})\sim(\hat{\pi}_i^\dagger,\hat{\pi}_{-i})}[-\mathcal{Z}_i] + \mathbb{E}_{\boldsymbol{a}\sim\hat{\pi}}[\mathcal{Z}_i]\Big] \nonumber\\
    &\quad + \tfrac{\eta}{2}\,\mathbb{E}_{x\sim\rho}\!\Big[\big\|Q_i^{\hat{\pi}_{-i}}(x,\cdot)-\hat{Q}_i(x,\cdot,\hat{\pi}_{-i})\big\|_\infty^2\Big].
\end{align}
Reorganizing the two linear $\mathcal{Z}_i$-terms into a single difference of evaluation distributions,
\begin{align}
\mathbb{E}_{x\sim\rho}\!\big[\delta_{\mathrm{BR},i}(x) + \delta_{\mathrm{Iter},i}(x)\big]
    &\;\le\; \underbrace{\mathbb{E}_{x\sim\rho}\!\Big[\mathbb{E}_{\boldsymbol{a}\sim\hat{\pi}}[\mathcal{Z}_i] - \mathbb{E}_{(a_i,\boldsymbol{a}_{-i})\sim(\hat{\pi}_i^\dagger,\hat{\pi}_{-i})}[\mathcal{Z}_i]\Big]}_{\text{linear mismatch}_i} \nonumber\\
    &\quad + \tfrac{\eta}{2}\,\mathbb{E}_{x\sim\rho}\!\Big[\big\|Q_i^{\hat{\pi}_{-i}}(x,\cdot)-\hat{Q}_i(x,\cdot,\hat{\pi}_{-i})\big\|_\infty^2\Big].
\label{eq:bias_cancellation}
\end{align}
The first bracket is the \emph{linear mismatch} of player $i$; the second is the \emph{squared evaluation error}, which we will bound through the concentrability shift in Section~\ref{app:conc-shift}.

\paragraph{Factoring the linear mismatch into an $L_1$ distance.} Because both expectations in $\text{linear mismatch}_i$ share the same opponent marginal $\hat{\pi}_{-i}$, only the player-$i$ marginal differs. Writing the difference explicitly,
\begin{align}
\mathbb{E}_{\boldsymbol{a}\sim\hat{\pi}}[\mathcal{Z}_i] - \mathbb{E}_{(a_i,\boldsymbol{a}_{-i})\sim(\hat{\pi}_i^\dagger,\hat{\pi}_{-i})}[\mathcal{Z}_i]
    &\;=\; \sum_{a_i}\big(\hat{\pi}_i(a_i|x) - \hat{\pi}_i^\dagger(a_i|x)\big)\,\mathcal{Z}_i(x,a_i,\hat{\pi}_{-i}) \nonumber\\
    &\;\le\; \big\|\hat{\pi}_i(\cdot|x) - \hat{\pi}_i^\dagger(\cdot|x)\big\|_1 \cdot \big\|\mathcal{Z}_i(x,\cdot,\hat{\pi}_{-i})\big\|_\infty \nonumber\\
    &\;\le\; \big\|\hat{\pi}_i(\cdot|x) - \hat{\pi}_i^\dagger(\cdot|x)\big\|_1 \cdot \|\mathcal{Z}_i\|_\infty,
\label{eq:linear_mismatch_l1}
\end{align}
where the first step uses H\"older's inequality and the second uses $\|\mathcal{Z}_i(x,\cdot,\pi_{-i})\|_\infty \le \|\mathcal{Z}_i\|_\infty$ for any marginal.

\paragraph{Two regimes for the linear mismatch.}
\begin{enumerate}
    \item \emph{ROPE.} When $\hat{\pi}$ is the exact regularized NE returned by Algorithm~\ref{alg:ideal_rope}, the first-order optimality of each $\hat{\pi}_i$ implies $\hat{\pi}_i(\cdot|x) \equiv \hat{\pi}_i^\dagger(\cdot|x)$ for every $i$ and $x$. Hence $\|\hat{\pi}_i-\hat{\pi}_i^\dagger\|_1 = 0$ pointwise, and the linear mismatch in \eqref{eq:bias_cancellation} \emph{vanishes identically}: the regression error $\mathcal{Z}_i$ cancels exactly between $\delta_{\mathrm{BR},i}$ and $\delta_{\mathrm{Iter},i}$.
    \item \emph{OPMD.} When $\hat{\pi}=\pi^{(t)}$ is an OPMD iterate with $t\in[T]$, $\pi_i^{(t)}\neq\pi_i^{\dagger,(t)}$ in general, but Corollary~\ref{cor:l1_optimization_rate} bounds the $L_1$ distance by a pure optimization quantity. Since the rewards and $\hat{Q}_i$ are bounded in $[0,1]$, $\|\mathcal{Z}_i\|_\infty\le 1$, and combining \eqref{eq:linear_mismatch_l1} with Corollary~\ref{cor:l1_optimization_rate},
\begin{align}
\mathbb{E}_{t^*}\!\Big[\sum_{i=1}^m \big|\text{linear mismatch}_i^{(t^*)}\big|\Big]
    &\;\le\; \mathbb{E}_{t^*}\!\Big[\sum_{i=1}^m \mathbb{E}_\rho\big[\|\pi_i^{(t^*)}-\pi_i^{\dagger,(t^*)}\|_1\big]\Big] \nonumber \\
    &\;\overset{\eqref{eq:l1_rate}}{\le}\; C_\nu\,\frac{\eta+\gamma}{\gamma}\,\mathcal{O}\!\Big(\frac{m}{\sqrt{T}}\Big)
    \;=\; \widetilde{\mathcal{O}}\!\Big(\frac{m}{\sqrt{T}}\Big).
    \label{eq:linear_mismatch_rate}
\end{align}
The linear mismatch therefore decays at the optimization rate $\widetilde{\mathcal{O}}(1/\sqrt{T})$, never with the dataset size $n$.
\end{enumerate}

\subsubsection{Concentrability shift to the offline distribution}
\label{app:conc-shift}

We now bound the squared evaluation error in \eqref{eq:bias_cancellation} by the in-sample squared regression error under the offline distribution $\mu$. The argument has three steps: (i) reduce the $\|\cdot\|_\infty$ over actions to a worst-case $\max_{a_i}$ via triangle and Jensen inequalities, (ii) shift the opponent profile from the iterate $\hat{\pi}_{-i}$ to the reference $\pi_{-i}^{\mathrm{ref}}$ at the cost of $C_{\mathrm{shift}}\coloneqq\exp(\eta(m-1))$, (iii) shift the resulting reference-anchored deviation distribution to $\mu$ at the cost of $C_{\mathrm{uni}}$ via Assumption~\ref{ass:concentrability}.

\paragraph{Step (i) --- Triangle and Jensen inequalities.} Using the marginalized identity \eqref{eq:Q_in_Z} and the triangle inequality applied at each $a_i$,
\begin{align}
\big\|Q_i^{\hat{\pi}_{-i}}(x,\cdot) - \hat{Q}_i(x,\cdot,\hat{\pi}_{-i})\big\|_\infty
    &\;=\; \max_{a_i}\,\Big|\mathbb{E}_{\boldsymbol{a}_{-i}\sim\hat{\pi}_{-i}}\!\big[-\mathcal{Z}_i(x,a_i,\boldsymbol{a}_{-i})\big]\Big| \nonumber \\
    &\;\le\; \max_{a_i}\,\mathbb{E}_{\boldsymbol{a}_{-i}\sim\hat{\pi}_{-i}}\!\big[|\mathcal{Z}_i(x,a_i,\boldsymbol{a}_{-i})|\big].
\end{align}
Squaring both sides and applying Jensen's inequality $(\mathbb{E} u)^2 \le \mathbb{E}[u^2]$ with $u=|\mathcal{Z}_i|$ inside the inner expectation,
\begin{align}
\big\|Q_i^{\hat{\pi}_{-i}}(x,\cdot)-\hat{Q}_i(x,\cdot,\hat{\pi}_{-i})\big\|_\infty^2
    &\;\le\; \Big(\max_{a_i}\,\mathbb{E}_{\boldsymbol{a}_{-i}\sim\hat{\pi}_{-i}}|\mathcal{Z}_i|\Big)^{\!2} \nonumber \\
    &\;\le\; \max_{a_i}\,\mathbb{E}_{\boldsymbol{a}_{-i}\sim\hat{\pi}_{-i}}\!\big[\mathcal{Z}_i(x,a_i,\boldsymbol{a}_{-i})^2\big].
\label{eq:triangle_jensen}
\end{align}
Taking expectations over $x\sim\rho$ and multiplying by $\eta/2$,
\begin{align}
\tfrac{\eta}{2}\,\mathbb{E}_{x\sim\rho}\!\Big[\big\|Q_i^{\hat{\pi}_{-i}}(\cdot)-\hat{Q}_i(\cdot,\hat{\pi}_{-i})\big\|_\infty^2\Big]
    \;\le\; \tfrac{\eta}{2}\,\mathbb{E}_{x\sim\rho}\!\Big[\max_{a_i}\,\mathbb{E}_{\boldsymbol{a}_{-i}\sim\hat{\pi}_{-i}}\!\big[\mathcal{Z}_i^2\big]\Big].
\label{eq:after_step1}
\end{align}

\paragraph{Step (ii) --- Opponent shift to the reference policy.} Each iterate $\pi_j^{(t)}$ takes a regularized Gibbs form anchored to $\pi_j^{\mathrm{ref}}$. By induction on $t$ from \eqref{eq:opmd_update}, there exists a function $g_j^{(t)}:\mathcal{X}\times\mathcal{A}_j\to\mathbb{R}$ that is a convex combination of empirical $Q$-values bounded in $[0,1]$ such that
\begin{align}
\pi_j^{(t)}(a_j|x) \;=\; \frac{\pi_j^{\mathrm{ref}}(a_j|x)\,\exp\!\big(\eta\,g_j^{(t)}(x,a_j)\big)}{\sum_{a'\in\mathcal{A}_j} \pi_j^{\mathrm{ref}}(a'|x)\,\exp(\eta\,g_j^{(t)}(x,a'))},\qquad g_j^{(t)}(x,a_j) \in [0,1].
\end{align}
The same form holds with $g_j = \bar{Q}_j$ for the regularized empirical NE returned by ROPE. Hence the per-player density ratio satisfies
\begin{align}
\frac{\pi_j^{(t)}(a_j|x)}{\pi_j^{\mathrm{ref}}(a_j|x)}
    \;=\; \frac{\exp(\eta\,g_j^{(t)}(x,a_j))}{\sum_{a'}\pi_j^{\mathrm{ref}}(a'|x)\,\exp(\eta\,g_j^{(t)}(x,a'))}
    \;\le\; \frac{\exp(\eta\,\max_a g_j^{(t)})}{\exp(\eta\,\min_a g_j^{(t)})}
    \;\le\; e^{\eta},
\end{align}
where the final inequality uses $g_j^{(t)}\in[0,1]$, so the numerator is at most $\exp(\eta\cdot 1)=e^\eta$ and the denominator (a convex combination of $\exp(\eta\,g_j^{(t)})\ge 1$) is at least $\exp(\eta\cdot 0)=1$.
Multiplying these per-player ratios across the $m-1$ opponents of player $i$ (using independence of marginals under the product policy),
\begin{align}
\frac{\hat{\pi}_{-i}(\boldsymbol{a}_{-i}|x)}{\pi_{-i}^{\mathrm{ref}}(\boldsymbol{a}_{-i}|x)}
    \;=\; \prod_{j\ne i}\frac{\hat{\pi}_j(a_j|x)}{\pi_j^{\mathrm{ref}}(a_j|x)}
    \;\le\; \exp\!\big(\eta(m-1)\big) \;\eqqcolon\; C_{\mathrm{shift}}.
\label{eq:Cshift_def}
\end{align}
Substituting \eqref{eq:Cshift_def} into the inner expectation in \eqref{eq:after_step1},
\begin{align}
\mathbb{E}_{\boldsymbol{a}_{-i}\sim\hat{\pi}_{-i}}\!\big[\mathcal{Z}_i(x,a_i,\boldsymbol{a}_{-i})^2\big]
    &\;=\; \sum_{\boldsymbol{a}_{-i}}\hat{\pi}_{-i}(\boldsymbol{a}_{-i}|x)\,\mathcal{Z}_i(x,a_i,\boldsymbol{a}_{-i})^2 \nonumber \\
    &\;\le\; C_{\mathrm{shift}}\,\sum_{\boldsymbol{a}_{-i}}\pi_{-i}^{\mathrm{ref}}(\boldsymbol{a}_{-i}|x)\,\mathcal{Z}_i(x,a_i,\boldsymbol{a}_{-i})^2 \nonumber \\
    &\;=\; C_{\mathrm{shift}}\,\mathbb{E}_{\boldsymbol{a}_{-i}\sim\pi_{-i}^{\mathrm{ref}}}\!\big[\mathcal{Z}_i(x,a_i,\boldsymbol{a}_{-i})^2\big],
\label{eq:opponent_shift}
\end{align}
which holds uniformly in $a_i$. Substituting \eqref{eq:opponent_shift} into \eqref{eq:after_step1},
\begin{align}
\tfrac{\eta}{2}\,\mathbb{E}_{x\sim\rho}\!\Big[\max_{a_i}\,\mathbb{E}_{\boldsymbol{a}_{-i}\sim\hat{\pi}_{-i}}\!\big[\mathcal{Z}_i^2\big]\Big]
    \;\le\; \tfrac{\eta\,C_{\mathrm{shift}}}{2}\,\mathbb{E}_{x\sim\rho}\!\Big[\max_{a_i}\,\mathbb{E}_{\boldsymbol{a}_{-i}\sim\pi_{-i}^{\mathrm{ref}}}\!\big[\mathcal{Z}_i^2\big]\Big].
\label{eq:after_step2}
\end{align}

\paragraph{Step (iii) --- Unilateral concentrability shift to $\mu$.} The remaining $\max_{a_i}$ in \eqref{eq:after_step2} is digested by introducing a deterministic deviation. Define
\begin{align}
\tilde{\pi}_i(\cdot|x) \;\coloneqq\; \delta_{a_i^\star(x)}(\cdot),
\qquad a_i^\star(x) \;\coloneqq\; \arg\max_{a_i}\,\mathbb{E}_{\boldsymbol{a}_{-i}\sim\pi_{-i}^{\mathrm{ref}}}\!\big[\mathcal{Z}_i(x,a_i,\boldsymbol{a}_{-i})^2\big],
\label{eq:tilde_pi_def}
\end{align}
and the joint deviation profile $\tilde{\pi} \coloneqq (\tilde{\pi}_i,\pi_{-i}^{\mathrm{ref}})$. By construction $\tilde{\pi}\in\Pi_{\mathrm{ref-uni}}$ (player $i$ deviates while the others adhere to the reference). With this construction the $\max_{a_i}$ becomes an expectation under $\tilde{\pi}_i$:
\begin{align}
\mathbb{E}_{x\sim\rho}\!\Big[\max_{a_i}\,\mathbb{E}_{\boldsymbol{a}_{-i}\sim\pi_{-i}^{\mathrm{ref}}}\!\big[\mathcal{Z}_i^2\big]\Big]
    &\;=\; \mathbb{E}_{x\sim\rho,\boldsymbol{a}\sim\tilde{\pi}(\cdot|x)}\!\big[\mathcal{Z}_i(x,\boldsymbol{a})^2\big] \nonumber \\
    &\;=\; \mathbb{E}_{(x,\boldsymbol{a})\sim\mu}\!\Big[\frac{\rho(x)\,\tilde{\pi}(\boldsymbol{a}|x)}{\mu(x,\boldsymbol{a})}\,\mathcal{Z}_i(x,\boldsymbol{a})^2\Big].
\label{eq:max_to_mu}
\end{align}
Applying Assumption~\ref{ass:concentrability} to the reference-anchored profile $\tilde{\pi}\in\Pi_{\mathrm{ref-uni}}$,
\begin{align}
\frac{\rho(x)\,\tilde{\pi}(\boldsymbol{a}|x)}{\mu(x,\boldsymbol{a})} \;\le\; C_{\mathrm{uni}}, \qquad \mu\text{-a.e.},
\end{align}
which substituted into \eqref{eq:max_to_mu} gives
\begin{align}
\mathbb{E}_{x\sim\rho}\!\Big[\max_{a_i}\,\mathbb{E}_{\boldsymbol{a}_{-i}\sim\pi_{-i}^{\mathrm{ref}}}\!\big[\mathcal{Z}_i^2\big]\Big]
    \;\le\; C_{\mathrm{uni}}\,\mathbb{E}_{(x,\boldsymbol{a})\sim\mu}\!\big[\mathcal{Z}_i(x,\boldsymbol{a})^2\big].
\label{eq:after_step3}
\end{align}
Combining \eqref{eq:after_step1}--\eqref{eq:after_step3}, the squared evaluation error in \eqref{eq:bias_cancellation} is bounded by the in-sample squared regression error under $\mu$:
\begin{align}
\tfrac{\eta}{2}\,\mathbb{E}_{x\sim\rho}\!\Big[\big\|Q_i^{\hat{\pi}_{-i}}(x,\cdot) - \hat{Q}_i(x,\cdot,\hat{\pi}_{-i})\big\|_\infty^2\Big]
    \;\le\; \tfrac{\eta\,C_{\mathrm{shift}}\,C_{\mathrm{uni}}}{2}\,\mathbb{E}_{(x,\boldsymbol{a})\sim\mu}\!\big[\mathcal{Z}_i(x,\boldsymbol{a})^2\big].
\label{eq:conc_shift}
\end{align}

\paragraph{Fast-rate generalization for least-squares.} The expected in-sample squared regression error is controlled by the standard fast-rate guarantee for empirical risk minimization with bounded outputs.

\begin{lemma}[Fast-rate generalization for least-squares]
\label{lem:mle_finite}
Let $\mathcal{Q}_i$ be a finite function class, and assume realizability $r_i\in\mathcal{Q}_i$ and bounded rewards $r_{\tau,i}\in[0,1]$. The least-squares estimator $\hat{Q}_i \in \arg\min_{f\in\mathcal{Q}_i}\sum_{\tau=1}^n (f(x_\tau,\boldsymbol{a}_\tau) - r_{\tau,i})^2$ satisfies, with probability at least $1-\delta$,
\begin{align}
    \mathbb{E}_{(x,\boldsymbol{a})\sim\mu}\!\big[\mathcal{Z}_i(x,\boldsymbol{a})^2\big]
    \;=\; \mathbb{E}_{(x,\boldsymbol{a})\sim\mu}\!\big[(\hat{Q}_i(x,\boldsymbol{a})-r_i(x,\boldsymbol{a}))^2\big]
    \;\le\; \frac{30\,\log(2|\mathcal{Q}_i|/\delta)}{n}.
\label{eq:mle_finite_bound}
\end{align}
\end{lemma}

Combining \eqref{eq:conc_shift} with Lemma~\ref{lem:mle_finite}, with probability at least $1-\delta$,
\begin{align}
\tfrac{\eta}{2}\,\mathbb{E}_{x\sim\rho}\!\Big[\big\|Q_i^{\hat{\pi}_{-i}}(x,\cdot) - \hat{Q}_i(x,\cdot,\hat{\pi}_{-i})\big\|_\infty^2\Big]
    \;\le\; \frac{15\,\eta\,C_{\mathrm{shift}}\,C_{\mathrm{uni}}\,\log(2|\mathcal{Q}_i|/\delta)}{n}.
\label{eq:per_player_squared}
\end{align}
This is the per-player squared evaluation error contribution; with $|\mathcal{Q}|\coloneqq\max_i|\mathcal{Q}_i|$, summing over $i\in[m]$ via a union bound (replacing $\delta\leftarrow\delta/m$) gives the total squared statistical residual $\widetilde{\mathcal{O}}\!\big(m\,\eta\,C_{\mathrm{shift}}\,C_{\mathrm{uni}}\,\log|\mathcal{Q}|/n\big)$.

\subsubsection{Closing the bound: idealized ROPE rate}
\label{app:close-rope}

Let $\hat{\pi}$ denote the policy returned by Algorithm~\ref{alg:ideal_rope}. Because $\hat{\pi}$ is an exact regularized NE of the empirical game, every $\hat{\pi}_i$ is the empirical regularized best response to $\hat{\pi}_{-i}$, so the potential ascent error and the linear mismatch both vanish identically:
\begin{align}
\varepsilon_{\mathrm{pot},i}(x) \;=\; \hat{V}_i^{\dagger,\hat{\pi}_{-i}}(x) - \hat{V}_i^{\hat{\pi}}(x) \;=\; 0,
\qquad
\hat{\pi}_i(\cdot|x) \;=\; \hat{\pi}_i^\dagger(\cdot|x),\quad\forall i,x.
\label{eq:rope_pot_vanishes}
\end{align}
Summing the per-iterate decomposition \eqref{eq:gap_decomp_named} over players, integrating over $x\sim\rho$, and using \eqref{eq:rope_pot_vanishes},
\begin{align}
\mathrm{Gap}_{\mathrm{NE}}(\hat{\pi})
    &\;=\; \sum_{i=1}^m \mathbb{E}_{x\sim\rho}\!\big[V_i^{\dagger,\hat{\pi}_{-i}}(x) - V_i^{\hat{\pi}}(x)\big] \nonumber\\
    &\;=\; \sum_{i=1}^m \mathbb{E}_{x\sim\rho}\!\big[\varepsilon_{\mathrm{pot},i}(x) + \delta_{\mathrm{BR},i}(x) + \delta_{\mathrm{Iter},i}(x)\big] \nonumber\\
    &\;\overset{\eqref{eq:rope_pot_vanishes}}{=}\; \sum_{i=1}^m \mathbb{E}_{x\sim\rho}\!\big[\delta_{\mathrm{BR},i}(x) + \delta_{\mathrm{Iter},i}(x)\big].
\label{eq:rope_chain_a}
\end{align}
Apply the bias-cancellation identity \eqref{eq:bias_cancellation} term by term. Since the linear mismatch is zero in the ROPE setting,
\begin{align}
\sum_{i=1}^m \mathbb{E}_{x\sim\rho}\!\big[\delta_{\mathrm{BR},i}(x) + \delta_{\mathrm{Iter},i}(x)\big]
    &\;\le\; \sum_{i=1}^m\Big\{\underbrace{\mathbb{E}_{x\sim\rho}[\text{linear mismatch}_i]}_{=\,0\text{ since }\hat{\pi}_i=\hat{\pi}_i^\dagger} \nonumber\\
    &\qquad\quad + \tfrac{\eta}{2}\,\mathbb{E}_{x\sim\rho}\!\Big[\big\|Q_i^{\hat{\pi}_{-i}}(x,\cdot)-\hat{Q}_i(x,\cdot,\hat{\pi}_{-i})\big\|_\infty^2\Big]\Big\} \nonumber\\
    &\;=\; \tfrac{\eta}{2}\sum_{i=1}^m \mathbb{E}_{x\sim\rho}\!\Big[\big\|Q_i^{\hat{\pi}_{-i}}(x,\cdot)-\hat{Q}_i(x,\cdot,\hat{\pi}_{-i})\big\|_\infty^2\Big].
\label{eq:rope_chain_b}
\end{align}
Apply the concentrability-shift bound \eqref{eq:conc_shift} and the fast-rate generalization Lemma~\ref{lem:mle_finite} to each summand:
\begin{align}
\tfrac{\eta}{2}\,\mathbb{E}_{x\sim\rho}\!\Big[\big\|Q_i^{\hat{\pi}_{-i}}(x,\cdot)-\hat{Q}_i(x,\cdot,\hat{\pi}_{-i})\big\|_\infty^2\Big]
    &\;\overset{\eqref{eq:conc_shift}}{\le}\; \tfrac{\eta\,C_{\mathrm{shift}}\,C_{\mathrm{uni}}}{2}\,\mathbb{E}_{(x,\boldsymbol{a})\sim\mu}\!\big[\mathcal{Z}_i(x,\boldsymbol{a})^2\big] \nonumber\\
    &\;\overset{\eqref{eq:mle_finite_bound}}{\le}\; \frac{15\,\eta\,C_{\mathrm{shift}}\,C_{\mathrm{uni}}\,\log(2|\mathcal{Q}_i|/\delta_i)}{n},
\label{eq:rope_chain_c}
\end{align}
each holding with probability $\ge 1-\delta_i$. Setting $\delta_i \coloneqq \delta/m$ and applying a union bound across $i\in[m]$, with probability at least $1-\delta$,
\begin{align}
\mathrm{Gap}_{\mathrm{NE}}(\hat{\pi})
    \;\le\; \sum_{i=1}^m \frac{15\,\eta\,C_{\mathrm{shift}}\,C_{\mathrm{uni}}\,\log(2m|\mathcal{Q}_i|/\delta)}{n}
    \;\le\; \frac{15\,m\,\eta\,C_{\mathrm{shift}}\,C_{\mathrm{uni}}\,\log(2m|\mathcal{Q}|/\delta)}{n},
\label{eq:rope_chain_d}
\end{align}
where $|\mathcal{Q}|\coloneqq\max_i|\mathcal{Q}_i|$. Substituting $C_{\mathrm{shift}}=\exp(\eta(m-1))$ and absorbing $\log(2m/\delta)$ into the $\widetilde{\mathcal{O}}$ notation,
\begin{align}
\mathrm{Gap}_{\mathrm{NE}}(\hat{\pi})
    \;\le\; \widetilde{\mathcal{O}}\!\left(\frac{m\,\eta\,C_{\mathrm{uni}}\,\exp(\eta(m-1))\,\log|\mathcal{Q}|}{n}\right),
\label{eq:rope_chain_final}
\end{align}
which is the fast $\widetilde{\mathcal{O}}(1/n)$ rate of Theorem~\ref{thm:fast_rate_rope}. Although the exact-NE oracle is intractable in generic general-sum games, Section~\ref{app:potential-bias} shows that the $\alpha$-potential structure aligns the equilibrium with the global maximization of $\hat{\Phi}$, motivating the practical OPMD analysis below.

\subsubsection{Closing the bound: practical OPMD rate}
\label{app:close-opmd}

We now turn to the iterate $\tilde{\pi} = \pi^{(t^*)}$ with $t^*\sim\mathrm{Unif}([T])$. Averaging the per-iterate decomposition \eqref{eq:gap_decomp_named} over $t^*$ and over players,
\begin{align}
\mathbb{E}_{t^*}\!\big[\mathrm{Gap}_{\mathrm{NE}}(\tilde{\pi})\big]
    \;=\; \frac{1}{T}\sum_{t=1}^T \sum_{i=1}^m \mathbb{E}_{x\sim\rho}\!\Big[\varepsilon_{\mathrm{pot},i}^{(t)}(x) + \delta_{\mathrm{BR},i}^{(t)}(x) + \delta_{\mathrm{Iter},i}^{(t)}(x)\Big].
\label{eq:opmd_avg_gap}
\end{align}
We bound the three contributions in turn.

\paragraph{Term 1 --- Potential ascent contribution.} By Lemma~\ref{lem:potential_ascent}, with stepsize $\gamma\le 1/L_\Phi$ and $L_\Phi=\mathcal{O}(m)$,
\begin{align}
\frac{1}{T}\sum_{t=1}^T\sum_{i=1}^m \mathbb{E}_{x\sim\rho}\!\big[\varepsilon_{\mathrm{pot},i}^{(t)}(x)\big]
    \;\le\; \mathcal{O}\!\Big(\frac{m}{\sqrt{T}}\Big) + m\alpha.
\label{eq:opmd_term_pot}
\end{align}

\paragraph{Term 2 --- Linear mismatch contribution.} Apply the bias-cancellation identity \eqref{eq:bias_cancellation} to each iterate and player. The two parts of the right-hand side give
\begin{align}
\frac{1}{T}\sum_{t=1}^T\sum_{i=1}^m \mathbb{E}_{x\sim\rho}\!\big[\delta_{\mathrm{BR},i}^{(t)}(x)+\delta_{\mathrm{Iter},i}^{(t)}(x)\big]
    &\;\le\; \frac{1}{T}\sum_{t,i}\mathbb{E}_{x\sim\rho}\!\big[|\text{linear mismatch}_i^{(t)}(x)|\big] \nonumber\\
    &\quad + \frac{\eta}{2T}\sum_{t,i}\mathbb{E}_{x\sim\rho}\!\Big[\big\|Q_i^{\pi_{-i}^{(t)}}(x,\cdot)-\hat{Q}_i(x,\cdot,\pi_{-i}^{(t)})\big\|_\infty^2\Big].
\label{eq:opmd_split}
\end{align}

The first sum is bounded by chaining \eqref{eq:linear_mismatch_l1}, Lemma~\ref{lem:l1_proportionality}, Cauchy--Schwarz, and the squared-step bound \eqref{eq:squared_step_bound}. Recalling $\|\mathcal{Z}_i\|_\infty\le 1$ (rewards and $\hat{Q}_i$ both lie in $[0,1]$),
\begin{align}
\frac{1}{T}\sum_{t,i}\mathbb{E}_\rho\big[|\text{linear mismatch}_i^{(t)}|\big]
    \;&\overset{\eqref{eq:linear_mismatch_l1}}{\le}\; \frac{1}{T}\sum_{t,i}\mathbb{E}_\rho\!\big[\|\pi_i^{(t)}-\pi_i^{\dagger,(t)}\|_1\big]\cdot \|\mathcal{Z}_i\|_\infty \nonumber\\
    &\;\overset{\eqref{eq:l1_proportionality}}{\le}\; \frac{C_\nu(\eta+\gamma)/\gamma}{T}\sum_{t,i}\mathbb{E}_\rho\!\big[\|\pi_i^{(t)}-\pi_i^{(t+1)}\|_1\big] \nonumber\\
    &\;\overset{\text{C.-S.}}{\le}\; \frac{C_\nu(\eta+\gamma)/\gamma}{T}\sqrt{Tm}\,\sqrt{\sum_{t,i}\mathbb{E}_\rho\!\big[\|\pi_i^{(t)}-\pi_i^{(t+1)}\|_1^2\big]} \nonumber\\
    &\;\overset{\eqref{eq:squared_step_bound}}{\le}\; C_\nu\,\frac{\eta+\gamma}{\gamma}\,\sqrt{\frac{m\cdot \mathcal{O}(m)}{T}}
    \;=\; \mathcal{O}\!\Big(C_\nu\,\frac{\eta+\gamma}{\gamma}\,\frac{m\sqrt{m}}{\sqrt{T}}\Big).
\label{eq:opmd_term_mismatch}
\end{align}
Treating $\eta,\gamma,\nu_{\min}$ as problem-dependent constants, this is $\widetilde{\mathcal{O}}(m\sqrt{m/T})$.

\paragraph{Term 3 --- Squared statistical residual contribution.} Applying the concentrability shift \eqref{eq:conc_shift} pointwise in $t$,
\begin{align}
\frac{\eta}{2T}\sum_{t=1}^T\sum_{i=1}^m\mathbb{E}_{x\sim\rho}\!\Big[\big\|Q_i^{\pi_{-i}^{(t)}}(x,\cdot)-\hat{Q}_i(x,\cdot,\pi_{-i}^{(t)})\big\|_\infty^2\Big]
    &\;\overset{\eqref{eq:conc_shift}}{\le}\; \frac{\eta\,C_{\mathrm{shift}}\,C_{\mathrm{uni}}}{2T}\sum_{t,i}\mathbb{E}_{(x,\boldsymbol{a})\sim\mu}\!\big[\mathcal{Z}_i^{(t)}(x,\boldsymbol{a})^2\big].
\end{align}
The pointwise estimator $\hat{Q}_i$ does not depend on $t$ (it is computed once from the offline dataset before any iteration), so $\mathcal{Z}_i^{(t)}\equiv \mathcal{Z}_i$. Applying Lemma~\ref{lem:mle_finite} per player with confidence $\delta/m$ and union-bounding across $i\in[m]$, with probability at least $1-\delta$,
\begin{align}
\frac{\eta\,C_{\mathrm{shift}}\,C_{\mathrm{uni}}}{2T}\sum_{t,i}\mathbb{E}_\mu\!\big[\mathcal{Z}_i(x,\boldsymbol{a})^2\big]
    &\;=\; \frac{\eta\,C_{\mathrm{shift}}\,C_{\mathrm{uni}}}{2}\sum_{i=1}^m \mathbb{E}_\mu\!\big[\mathcal{Z}_i^2\big] \nonumber\\
    &\;\overset{\eqref{eq:mle_finite_bound}}{\le}\; \frac{15\,m\,\eta\,C_{\mathrm{shift}}\,C_{\mathrm{uni}}\,\log(2m|\mathcal{Q}|/\delta)}{n}
    \;=\; \widetilde{\mathcal{O}}\!\Big(\frac{m\,\eta\,C_{\mathrm{shift}}\,C_{\mathrm{uni}}\,\log|\mathcal{Q}|}{n}\Big).
\label{eq:opmd_term_stat}
\end{align}

\paragraph{Final synthesis.} Substituting \eqref{eq:opmd_term_pot}, \eqref{eq:opmd_term_mismatch}, and \eqref{eq:opmd_term_stat} into \eqref{eq:opmd_avg_gap} via \eqref{eq:opmd_split},
\begin{align}
\mathbb{E}_{t^*}\!\big[\mathrm{Gap}_{\mathrm{NE}}(\tilde{\pi})\big]
    &\;\le\; \underbrace{\mathcal{O}\!\Big(\frac{m}{\sqrt{T}}\Big)}_{\text{potential ascent }(\varepsilon_{\mathrm{pot}})}
    + \underbrace{\widetilde{\mathcal{O}}\!\Big(\frac{m\sqrt{m}}{\sqrt{T}}\Big)}_{\text{geometric mismatch }(\delta_{\mathrm{BR}}+\delta_{\mathrm{Iter}})} \nonumber\\
    &\quad + \underbrace{m\alpha}_{\text{bias floor}}
    + \underbrace{\widetilde{\mathcal{O}}\!\Big(\frac{m\,\eta\,C_{\mathrm{shift}}\,C_{\mathrm{uni}}\,\log|\mathcal{Q}|}{n}\Big)}_{\text{squared statistical residual}}.
\label{eq:opmd_final_decomp}
\end{align}
Both optimization rates decay as $\mathcal{O}(T^{-1/2})$. Setting $T\ge n^2$ makes the optimization terms strictly dominated by the statistical residual: $\mathcal{O}(m/\sqrt{T}) + \widetilde{\mathcal{O}}(m\sqrt{m}/\sqrt{T}) \le \widetilde{\mathcal{O}}(1/n)$. Substituting into \eqref{eq:opmd_final_decomp},
\begin{align}
\mathbb{E}_{t^*}\!\big[\mathrm{Gap}_{\mathrm{NE}}(\tilde{\pi})\big]
    \;\le\; \widetilde{\mathcal{O}}\!\Big(\frac{1}{n}\Big) + m\alpha,
\label{eq:opmd_one_over_n}
\end{align}
with hidden constant scaling as $m\,\eta\,C_{\mathrm{shift}}\,C_{\mathrm{uni}}\,\log|\mathcal{Q}|$ and $C_{\mathrm{shift}}=\exp(\eta(m-1))$. The exponential dependence on $\eta(m-1)$ is the price of converting opponent-policy density ratios via the Gibbs form of the iterates and is characteristic of KL-regularized analyses (cf.\ \citep{ye2024online}). This proves Theorem~\ref{thm:opmd_rate}.



\end{document}